%% file: VAEhmm_arXiv.tex
\documentclass[nolayout]{article} 
\input{math_commands.tex}

\usepackage{hyperref}
\usepackage{url}
\usepackage{amssymb}
\usepackage{dsfont}
\usepackage{amsthm}

\newcommand{\filtmeas}{\Phi}
\newcommand{\filtdens}{\phi}

\newcommand{\XinitIS}[2][]
{\ifthenelse{\equal{#1}{}}{\ensuremath{\rho_{#2}}}{\ensuremath{\check{\rho}_{#2}}}}

\newcommand{\rmd}{\ensuremath{\mathrm{d}}}

\newcommand{\eqsp}{\;}

\newcommand{\filt}[2][]%
{%
\ifthenelse{\equal{#1}{}}{\ensuremath{\phi_{#2}}}{\ensuremath{\phi_{#1,#2}}}%
}

\newcommand{\sumwght}[2][]{%
\ifthenelse{\equal{#1}{}}{\ensuremath{\Omega_{#2}}}{\ensuremath{\Omega_{#2}^{(#1)}}}}
\newcommand{\sumwghthat}[2][]{%
\ifthenelse{\equal{#1}{}}{\ensuremath{\widehat{\Omega}_{#2}}}{\ensuremath{\widehat{\Omega}_{#2}^{(#1)}}}}

\newcommand{\udlow}{\sigma_-}
\newcommand{\udup}{\sigma_+}
\newcommand{\udlowvar}{\vartheta_-}
\newcommand{\udupvar}{\vartheta_+}

\newcommand{\parvec}{\theta}
\newcommand{\parspace}{\Theta}

\newcommand{\parvar}{\varphi}
\newcommand{\parvarspace}{\Phi}
\newcommand{\vard}[1]{q_{\parvar,#1}} 
\newcommand{\Xset}{\mathsf{X}}
\def\Yset{\mathsf{Y}}
\newcommand{\dimt}{d_\theta}
\newcommand{\dimv}{d_\varphi}

\newcounter{hypH}
\newenvironment{hypH}{\refstepcounter{hypH}\begin{itemize}
\item[{\bf H\arabic{hypH}}]}{\end{itemize}}

\newcounter{hypA}
\newenvironment{hypA}{\refstepcounter{hypA}\begin{itemize}
\item[{\bf A\arabic{hypA}}]}{\end{itemize}}

\newtheorem{theorem}{Theorem}[section]
\newtheorem{lemma}[theorem]{Lemma}
\newtheorem{proposition}[theorem]{Proposition}
\newtheorem{corollary}[theorem]{Corollary}

\title{Variational excess risk bound  for general state space models}

\author[*]{\'Elisabeth Gassiat}
\author[$\dag$]{Sylvain Le Corff}

\affil[*]{{\small  Universit\'e Paris-Saclay, CNRS. Laboratoire de mathématiques d’Orsay, 91405, Orsay, France.}}
\affil[$\dag$]{{\small LPSM, 
       Sorbonne Universit\'e, UMR CNRS 8001, Paris, France.}}

\begin{document}

\maketitle

\begin{abstract}
 In this paper, we consider variational autoencoders (VAE) for general state space models. 
We consider a backward factorization of the variational distributions to analyze the excess risk associated with VAE. Such backward factorizations were recently proposed to perform online variational learning \cite{campbell2021online} and to obtain upper bounds on the variational estimation error \cite{mathisJMLR}. When independent trajectories of sequences are observed and under strong mixing assumptions on the state space model and on the variational distribution, we provide an oracle inequality explicit in the number of samples and in the length of the observation sequences. 
We then derive consequences of this theoretical result. In particular, when the data distribution is given by a state space model, we provide an upper bound for the Kullback-Leibler divergence between the data distribution and its estimator and between the variational posterior and the estimated state space posterior distributions.
Under classical assumptions, we prove that our results can be applied to Gaussian backward kernels built with dense and recurrent neural networks.
\end{abstract}

\section{Introduction}

Deep generative models have been increasingly used and analyzed for the past few years. In this setting, Variational autoencoders (VAEs) offer the possibility to simultaneously model and train (i) the conditional distribution of the observation given latent variables referred to as the decoder, and (ii) a variational approximation of the conditional distribution of the latent variable given the observation referred to as the encoder. They have been successfully applied in many contexts such as image generation
\cite{vahdat2020nvae}, text generation \cite{bowman2015generating}, state estimation and image reconstruction \cite{pmlr-v162-cohen22b}.

Variational inference has been widely and satisfactorily used for many practical applications but its theoretical properties has been analyzed only very recently.
Theoretical guarantees have been mostly proposed for variational inference procedures in settings where datasets are based on independent data and for mean-field approximations. In \cite{huggins2020validated}, the authors provided variational error bounds, in particular for the estimation of the posterior mean and covariance. In \cite{10.1214/18-EJS1475}, the authors established the concentration of variational approximations of posterior distributions for mixtures of general laws using PAC-Bayesian theory. The PAC-Bayesian theory has also been used 
in \cite{mbackeEtal2023} where the authors controlled in particular the $\mathrm{L}^2$ reconstruction loss under the true data distribution for VAEs. In addition, \cite{tang21a} provided a theoretical analysis of the excess risk for Empirical Bayes Variational Autoencoders for both parametric and nonparametric settings. They derived a set of generic assumptions to obtain an oracle inequality explicit in the number of samples and  proposed an upper bound for the total variation distance
between the true distribution of the observations and a variational approximation combining the empirical distribution of the dataset and the proposed VAE architecture.

In this paper, we aim at extending the theoretical results on variational inference procedures in two directions. First, we set the focus on the use of VAEs for general state space models, i.e. settings where the 
decoding distribution $P^Y_{\theta}$ of the observations depends on an unobserved Markov chain.  
In addition, instead of using mean-field approximations, we consider variational encoding  distributions $Q_{\varphi}$ satisfying a backward factorization as proposed in \cite{campbell2021online,mathisJMLR}. In \cite{mathisJMLR}, the authors derived the first theoretical results providing upper bounds on the state decoding estimation error when using variational inference with backward factorization and no such results were proposed for state space models using a mean-field approximations. This factorization was used in \cite{campbell2021online} to define new online variational estimation algorithms, where observations are processed on-the-fly.

In this paper, we provide the first (up to our knowledge) theoretical guarantees on the trained variational approximation in the setting of independent copies of sequences with distribution $P_{\cal D}$ when using a backward variational factorization. 
\begin{itemize}
    \item 
We
provide assumptions on the 
decoding
and variational encoding kernels  under which we prove an oracle inequality for the risk explicit in particular in the number of samples and in the length of the observation sequences, see Theorem~\ref{th:oracle}. This result is established using an alternative formulation of \cite[Theorem 3]{tang21a} in our state space setting and with an explicit dependency on  some constants to track all terms depending on the number of observations. This allows to understand when the procedure leads to a decoding distribution that approximates well the data distribution together with a coding distribution which approximates well the decoding state distribution.
\item 
In particular, when data are generated from a general state space model, and when $P_{\cal D}$ belongs to the decoding family of distributions, we give an upper bound also explicit in the way the backward coding kernels approximate the backward decoding kernels, see Corollary \ref{cor:kl}.
\item
We analyse settings in which our results hold, in particular settings with Gaussian backward kernels based on Multi-Layer Perceptrons (MLPs) and on Recurrent Neural Networks (RNNs).
\end{itemize}
The paper  is organised as follows. The general setting and notations for state space models and variational learning are given in Section~\ref{sec:setting}. Assumptions and theoretical results are proposed in Section~\ref{sec:results} along with discussions on specific deep architectures used in practice. A discussion with insights for future works is given in Section~\ref{sec:discussion}. Detailed proofs of theoretical results are given in Appendices \ref{app:A} and \ref{app:B}. Additional proofs to highlight that when the state and observation spaces are compact our main results hold are given in  Appendix \ref{app:C}.

\section{State space model and variational estimation}
\label{sec:setting}

Let $\Theta\subset \mathbb{R}^{\dimt}$ be a parameter space. In this paper, we consider a  state-space model depending on $\parvec\in\parspace$, i.e. a bivariate discrete-time process $\{(X_t,Y_t)\}_{t\geq 0}$ where $\{X_t\}_{t\geq 0}$ is a hidden Markov chain  in a measurable space $(\Xset,\mathcal{X})$ with initial distribution $\chi$ with density $\zeta$ with respect to a reference measure $\mu$ and for all $t \geqslant 0$, the conditional distribution of $X_{t+1} $ given $X_{0:t}$ is written $M_\theta(X_t,\cdot)$ and  has density $m_{\parvec}(X_{t},\cdot)$, where $a_{u:v}$ is a short-hand notation for $(a_u,\ldots,a_v)$ for $0\leqslant u \leqslant v$ and any sequence $(a_\ell)_{\ell\geqslant 0}$. The observations $\{Y_t\}_{0\leqslant t \leqslant T}$ take values in a measurable space $(\Yset,\mathcal{Y})$ and they are assumed to be independent conditionally on $X_{0:T}$ and, for all $0\leqslant t \leqslant T$, the distribution of $Y_t$ given $X_{0:T}$ depends on $X_t$ only, is written $G_{\parvec}(X_t,\cdot)$, and has density $y\mapsto g^y_{\parvec}(X_t)$ with respect to a reference measure $\nu$. 

In this context, the joint probability distribution $P_{\theta}$ of $(X_{0:T},Y_{0:T})$ has density with respect to $\mu^{\otimes (T+1)}\otimes \nu^{\otimes (T+1)}$ given, for all $\theta\in\Theta$, $x_{0:T}\in \Xset^{T+1}$ and all $y_{0:T}\in\Yset^{T+1}$, by
$$
p_{\theta,0:T}^{}(x_{0:T},y_{0:T}) =  \zeta(x_{0})g_\theta^{y_0}(x_0) \prod_{t=1}^Tm_\theta(x_{t-1},x_t)g_\theta^{y_t}(x_t)\eqsp,
$$
and the joint smoothing distribution, i.e. the conditional distribution of $X_{0:T}$ given $Y_{0:T}$, is given for all measurable function $h$ by 
$$
\filtmeas^{y_{0:T}}_{\theta,0:T|T}(h)  = \frac{\int\chi(\rmd x_{0})g_\theta^{y_0}(x_0) \prod_{t=1}^TM_\theta(x_{t-1},\rmd x_t)g_\theta^{y_t}(x_t)h(x_{0:T})}{\int\chi(\rmd x_{0})g_\theta^{y_0}(x_0) \prod_{t=1}^TM_\theta(x_{t-1},\rmd x_t)g_\theta^{y_t}(x_t)}.
$$
The probability density of $\filtmeas^{y_{0:T}}_{\theta,0:T|T}$ is denoted by $\filtdens^{y_{0:T}}_{\theta,0:T|T}$.
In the following, we use the notation $\Phi^{y_{0:t}}_{\theta,t} = \Phi^{y_{0:t}}_{\theta,0:t|t}$ to denote the the filtering distribution at time $t$, i.e. the conditional distribution of $X_t$ given $Y_{0:t}$, with a similar convention for the probability densities. The joint smoothing distribution can also be written
$$
\filtmeas^{y_{0:T}}_{\theta,0:T|T}(\rmd x_{0:T}) = \filtmeas^{y_{0:T}}_{\theta,T}(\rmd x_{T}) \prod_{t=0}^{T-1} B^{y_{0:T-t-1}}_{\theta,T-t-1|T-t}(x_{T-t},\rmd x_{T-t-1})\,,
$$
where $B^{y_{0:T-t-1}}_{\theta,T-t-1|T-t}(x_{T-t},\rmd x_{T-t-1})$ is the backward kernel at time $T-t$ defined by 
$$
B^{y_{0:T-t-1}}_{\theta,T-t-1|t}(x_{T-t},\rmd x_{T-t-1})\propto \filtmeas^{y_{0:T-t-1}}_{\theta,T-t-1}(\rmd x_{T-t-1})m_\theta(x_{T-t-1},x_{T-t})\eqsp,
$$ 
with a probability density with respect to $\mu$ denoted by $b^{y_{0:T-t-1}}_{\theta,T-t-1|T-t}(x_{T-t},\cdot)$. 
 For all $T$, $\theta$, $y_{0:T}\in\Yset^{T+1}$, the loglikelihood of the observations is:
$$
\ell^{y_{0:T}}_{T}(\theta) =  \log L^{y_{0:T}}_{T}(\theta)\eqsp,
$$
where
$$
L^{y_{0:T}}_{T}(\theta) =\int p_{\theta,0:T}(x_{0:T},y_{0:T}) \mu(\rmd x_{0:T})\eqsp.
$$
The joint smoothing distribution is usually intractable and we focus in this paper on variational learning to perform approximate maximum likelihood. Following  \cite{campbell2021online,mathisJMLR}, we propose a backward variational formulation:
$$
Q_{\varphi,0:T}^{y_{0:T}}(\rmd x_{0:T}) = Q_{\varphi,T}^{y_{0:T}}(\rmd x_T)\prod_{t=0}^{T-1} Q^{y_{0:T}}_{\varphi,T-t-1|T-t}(x_{T-t},\rmd x_{T-t-1})\,,
$$
where $\varphi\in\Phi\subset \mathbb{R}^{\dimv}$, and where $Q^{y_{0:T}}_{\varphi,T-t-1|T-t}(x_{T-t},\cdot)$ (resp. $Q_{\varphi,T}^{y_{0:T}}$) has probability density $q^{y_{0:T}}_{\varphi,T-t-1|T-t}(x_{T-t},\cdot)$ (resp. $q_{\varphi,T}^{y_{0:T}}$) with respect to the reference measure $\mu$.
In this setting, the ELBO writes, for all $\parvec\in\parspace$, $\parvar\in\parvarspace$, and for a sequence of observations $Y_{0:T}$,
$$
\mathrm{ELBO}^{Y_{0:T}}_T(\theta,\varphi) 
= \ell^{Y_{0:T}}_{T}(\theta)- \mathrm{KL}\left(Q^{Y_{0:T}}_{\varphi,0:T}\middle\| \filtmeas_{\theta,0:T|T}^{Y_{0:T}}\right)\,.
$$
Let $(Y^i_{0:T})_{1\leq i \leq n}$ be i.i.d. sequences with distribution $P_{\mathcal{D}}$ with density $p_{\mathcal{D}}$. Maximizing 
$(\theta,\varphi)\mapsto\sum_{i=1}^{n}\mathrm{ELBO}^{Y^i_{0:T}}_T(\theta,\varphi)$
is equivalent to minimizing the following loss function 
$$
\mathcal{L}_{n,T}(\theta,\varphi) = \frac{1}{n}\sum_{i=1}^nm(\theta,\varphi,Y^i_{0:T})\,,
$$
where
$$
  m(\theta,\varphi,Y^i_{0:T}) = \log \frac{p_{\mathcal{D}}(Y^i_{0:T})}{L^{Y^i_{0:T}}_{T}(\theta)} 
  + \mathrm{KL}\left(Q^{Y^i_{0:T}}_{\varphi,0:T}\middle\| \filtmeas_{\theta,0:T|T}^{Y^i_{0:T}}\right)\eqsp.
  $$
Define
$$
(\widehat\theta_{n,T},\widehat\varphi_{n,T}) \in\mathrm{argmin}_{\theta\in\parspace,\varphi\in\parvarspace} \;\mathcal{L}_{n,T}(\theta,\varphi)\eqsp.
$$
Such a procedure is a so-called $M$-estimation method in the statistical literature. The intuition is that with large data sets, that is when $n$ is large, the ELBO is closed to the expected  value of $m$ under the unknown distribution of the data, and the estimated decoding and coding parameters are close to minimize this expected value. An important body of work in the statistical community has been devoted to develop very general settings in which non asymptotic bounds on the risk of $M$-estimators, referred to as  oracle inequalities, can be given, see \cite{saraM} as early reference, or \cite{MR3967104} and the references therein  for more recent results. Moreover, oracle inequalities are obviously the only property one can hope for such estimators, the other properties being consequences of the oracle inequality.
In the following section, we thus first provide assumptions under which we obtain an oracle inequality and then discuss consequences.

\section{Main results}
\label{sec:results}

\subsection{Notations. }
In the following, for all measures $\lambda$ and $\eta$ on $(\Xset,\mathcal{X})$ and all transition kernels 
$K$ 
we consider the following notations. For all measurable sets $A\subset \Xset \times \Xset$,
$\lambda\otimes \eta(A) = \int \mathds{1}_A(x,x')\lambda(\rmd x)\eta(\rmd x')$
and $\lambda\otimes K(A) =  \int \mathds{1}_A(x,x')\lambda(\rmd x)K(x,\rmd x')$, 
for all measurable sets
$B\subset \Xset$,  $\lambda K(B) =  \int \lambda(\rmd x)\mathds{1}_B(x')K(x,\rmd x')$, and for all  real-valued measurable functions $h$ on $(\Xset,\mathcal{X})$, $\lambda(h) = \int\lambda(\rmd x)h(x)$. For all measurable functions $h_1, h_2$, we write $h_1\otimes h_2 :(x,x')\mapsto h_1(x) h_2(x')$. For all $\alpha>0$, define on $\mathbb{R}_+$ the function $\psi_{\alpha}: x\mapsto \exp(x^\alpha)-1$.  For all real-valued random variables $X$, define the Orlicz norm of order $\alpha$ 
by
$$
\|X\|_{\psi_\alpha} = \mathrm{inf}_{\lambda>0}\left\{\mathbb{E}\left[\psi_\alpha(|X|/\lambda)\right]\leq 1\right\}\eqsp.
$$
For all probability measures $P$ and $Q$ defined on the same probability space, $\|P-Q\|_{\mathrm{tv}}$ will denote the total variation norm between $P$ and $Q$, and $\mathrm{KL}\left(Q\middle\| P\right)$ their Kullback-Leibler divergence, that is $\mathrm{KL}\left(Q\middle\| P\right)=\E_{Q}[\log (\rmd Q/\rmd P)]$.

\subsection{Assumptions}
In this section, we propose a set of 
assumptions on the kernel densities $m_\theta$ and $\vard{t\vert t+1}^{y_{0:T}}$, $0\leq t\leq T-1$, and on the conditional densities $g_{\theta}^{y}$, under which we are able to prove
an oracle inequality. In the state space model literature, Assumption H\ref{assum:strong:mixing} is usual 
to control smoothing expectations and  H\ref{assum:bound:likelihood} for the study of asymptotic properties of maximum likelihood estimators. More assumptions are needed to manage the complexity of the models and to get a nonasymptotic control of the risk of the estimators. These controls are obtained with Assumptions H\ref{assum:lip}-\ref{assum:pdata}. We discuss in Section \ref{sec:appli} how they can be applied to specific architectures used in practice.
Additional discussions on the assumptions are provided in Appendix~\ref{sec:check:assum} where we prove that usual compact state space models are covered by our theory. 

\begin{hypH}
\label{assum:strong:mixing}
There exist probability measures $\eta_-$ and $\eta_+$ on $(\Xset,\mathcal{X})$ and constants $0 < \udlow < \udup < \infty$ such that for all $\parvec\in\Theta$, $x\in \Xset$, all measurable set $A$,
$$
\udlow \eta_-(A)\leq \chi(A) \leq \udup \eta_+(A)
$$
and
$$
\udlow\eta_-(A) \leq M_{\parvec}(x, A) \leq \udup \eta_+ (A)\eqsp.
$$ 
There exist probability measures $\lambda_-$ and $\lambda_+$ on $(\Xset,\mathcal{X})$ such that for all $y_{0:T}\in\Yset^{T+1}$, there exist $\udlowvar^{y_{0:T}}>0$ and $\udupvar^{y_{0:T}}>0$ such that for all $\varphi\in\Phi$, $t\geq 0$, $x\in \Xset$, all measurable set $A$,
$$ \udlowvar^{y_{0:T}}\lambda_-(A) \leq  Q_{\varphi,t\vert t+1}^{y_{0:T}}(x, A) \leq \udupvar^{y_{0:T}}\lambda_+(A)\eqsp. 
$$ 
In addition, for all $\varphi\in\Phi$, all $y_{0:T}\in \Yset^{T+1}$, and all measurable set $A$,
$$
\udlowvar^{y_{0:T}}\lambda_-(A) \leq  Q_{\varphi,T}^{y_{0:T}}(A) \leq \udupvar^{y_{0:T}}\lambda_+(A).
$$
\end{hypH}

\begin{hypH}
\label{assum:bound:likelihood}
For all $y\in\Yset$, $\mathrm{inf}_{\theta\in\Theta}\int g_\theta^y(x)\eta_-(\rmd x) = c_-(y)>0$ and $\mathrm{sup}_{\theta\in\Theta}\int g_\theta^y(x) \eta_+(\rmd x) = c_+(y)<\infty$. 

\end{hypH}
We consider also the following notation $\mathrm{sup}_{\theta\in\Theta} g_\theta^{y_t} = \bar g^{y_t}$ and $\mathrm{inf}_{\theta\in\Theta} g_\theta^{y_t} = \underline g^{y_t}$. We constrain the kernels and the conditional densities to be Lipschitz in the parameters with a Lipschitz coefficient depending on the variables.
\begin{hypH}
\label{assum:lip}
There exists $M$ such that for all $\theta, \theta'\in\Theta$ and $x,x'\in\Xset$,
$$
\left| m_{\theta}(x,x') - m_{\theta'}(x,x')\right| \leq  M(x,x')\|\theta-\theta'\|_2\eqsp.
$$
For all $1\leq t \leq T$, $y_{0:T}\in\Yset^{T+1}$, there exists $K^{y_{0:T}}_{t-1|t}$ such that for all $\varphi, \varphi'\in\Phi$ and $x,x'\in\Xset$,
$$
    \left| q^{y_{0:T}}_{\varphi,t-1|t}(x,x') - q^{y_{0:T}}_{\varphi',t-1|t}(x,x')\right| 
    \leq  K^{y_{0:T}}_{t-1|t}(x',x)\|\varphi-\varphi'\|_2\eqsp.
$$
In addition, there exists $K^{y_{0:T}}_{T}$ such that for all $\varphi, \varphi'\in\Phi$ and $x\in\Xset$,
$$
\left| q^{y_{0:T}}_{\varphi,T}(x) - q^{y_{0:T}}_{\varphi',T}(x)\right| \leq  K^{y_{0:T}}_{T}(x)\|\varphi-\varphi'\|_2\eqsp.
$$
For all 
$y\in\Yset$, there exists $G^{y}$ such that for all $\theta, \theta'\in\Theta$ and $x\in\Xset$,
$$
\left| g_\theta^{y}(x) - g_{\theta'}^{y}(x)\right| \leq  G^{y}(x)\|\theta-\theta'\|_2\eqsp.
$$
\end{hypH} 
Define, for $1\leq t \leq T-1$, 
\begin{equation}
\label{eq:def:addfunc}
h^{y_{0:T}}_{t,\theta,\varphi}(x_{t-1},x_t) = \log q^{y_{0:T}}_{\varphi,t-1|t}(x_t,x_{t-1}) 
 - \log b^{y_{0:t-1}}_{\theta,t-1|t}(x_t,x_{t-1})
\end{equation} 
and, by convention, $h^{y_{0:T}}_{T,\theta,\varphi}(x_{T-1},x_T) = \log q^{y_{0:T}}_{\varphi,T-1|T}(x_T,x_{T-1}) - \log b^{y_{0:T-1}}_{\theta,T-1|T}(x_T,x_{T-1}) + \log q^{y_{0:T}}_{\varphi,T}(x_T) - \log \phi^{y_{0:T}}_{\theta,T}(x_T)$.
\begin{hypH}
\label{assum:boundH}
For all $y_{0:T}\in \Yset^{T+1}$ and all $0\leq t\leq T$,
$$
\sup_{\theta\in\Theta,\varphi\in\Phi}\left\|\int \lambda_+(\rmd x)\left|h^{y_{0:T}}_{t,\theta,\varphi}(x,\cdot)\right|\right\|_\infty = \upsilon_t^{y_{0:T}}<\infty\eqsp,
$$
and for all $\theta,\theta'\in\Theta$, $\varphi,\varphi'\in\Phi$, $1\leq t\leq T$,
\begin{align*}
    \int \lambda_+\otimes\lambda_+(\rmd x\rmd x')\left|\log q^{y_{0:T}}_{\varphi,t-1|t}(x,x') - \log q^{y_{0:T}}_{\varphi',t-1|t}(x,x')\right| &\leq c_{1,t}^{y_{0:T}}\left\|\varphi-\varphi'\right\|_2\eqsp, \\
    \int \lambda_+\otimes\lambda_+(\rmd x\rmd x')\left|\log b^{y_{0:t-1}}_{\theta,t-1|t}(x,x') - \log b^{y_{0:t-1}}_{\theta',t-1|t}(x,x')\right| &\leq c_{2,t}^{y_{0:t-1}}\left\|\theta-\theta'\right\|_2\eqsp,\\
    \int \lambda_+(\rmd x)\left|\log q^{y_{0:T}}_{\varphi,T}(x) - \log q^{y_{0:T}}_{\varphi',T}(x)\right| &\leq c_{3,T}^{y_{0:T}}\left\|\varphi-\varphi'\right\|_2\eqsp,\\
    \int \lambda_+(\rmd x)\left|\log \phi^{y_{0:T}}_{\theta,T}(x) - \log \phi^{y_{0:T}}_{\theta',T}(x)\right| &\leq c_{4,T}^{y_{0:T}}\left\|\theta-\theta'\right\|_2\eqsp,
\end{align*}
where $\lambda_+$ is defined in H\ref{assum:strong:mixing}.
\end{hypH}
Our upper bounds require to prove that $m$ is a Lipschitz function of the parameters, and we need an upper bound on the $\mathrm{L}^2$-norm of the Lipschitz coefficient. For this, we consider the following moment assumptions. 
\begin{hypH}
\label{assum:moments}
There exists $A$ such that  the following inequalities are satisfied.
$$
\E\left[\left(\udupvar^{Y_{0:T}}c_{3,T}^{Y_{0:T}}\right)^2\right]\leq A\eqsp,\eqsp
\E\left[\left(\udupvar^{Y_{0:T}}c_{4,T}^{Y_{0:T}}\right)^2\right]\leq A\eqsp,
$$
for all $0\leq t\leq T$,
$$
\mathbb{E}\left[\frac{\mu(G^{Y_t})^2}{c_-(Y_t)^2}\right]\leq A\eqsp,\eqsp
\E\left[\left((\udupvar^{Y_{0:T}})^2 c_{1,t}^{Y_{0:T}}\right)^2\right]\leq A\eqsp,
$$
for all $1\leq t\leq T$,
$$
\E\left[\left((\udupvar^{Y_{0:T}})^2 c_{2,t}^{Y_{0:t-1}}\right)^2\right]\leq A\eqsp,\eqsp
 \mathbb{E}\left[\frac{\eta_+\otimes\mu(M\otimes  \bar g^{Y_{t-1}}\bar g^{Y_t})^2}{c_-(Y_{t-1})^2c_-(Y_{t})^2}\right]\leq A\eqsp,
 $$
 $$
\mathbb{E}\left[\left(\udupvar^{Y_{0:T}}\sum_{s=t-1}^T \lambda_+\otimes \lambda_+   (K^{y_{0:T}}_{s|s+1}) \rho(Y_{0:T})^{s-t}\right)^2\right]\leq A\eqsp,
$$
where for all $y_{0:T}$, $\rho(y_{0:T}) = 1-\udlowvar^{y_{0:T}}$,
for all $0\leq s,t\leq T$,
$$
\mathbb{E}\left[\frac{c_+(Y_t)^2\mu(G^{Y_s})^2}{c_-(Y_t)^2c_-(Y_s)^2}\right]\leq A\eqsp,
$$
and for all $0\leq t\leq T$, all $1\leq s\leq T$,
$$
 \mathbb{E}\left[\left(\frac{c_+(Y_t)\eta_+\otimes\mu(M\otimes  \bar g^{Y_{s-1}}\bar g^{Y_s})}{c_-(Y_{s-1})c_-(Y_{s})c_-(Y_{t})}\right)^2\right]\leq A\eqsp.
$$
\end{hypH}
The following assumption is used to have concentration properties, as usual in the statistical literature to get theoretical guarantees with finite samples.

\begin{hypH}
    \label{assum:pdata}
    There exists $\alpha_*$ and $B>0$ such that for all $T\geq 1$,
    $$
\left\|\log p_{\mathcal{D}}(Y_{0:T}) \right\|_{\psi_{\alpha_*}}\leq B T\quad \mathrm{and}\quad 
\left\|(\udupvar^{Y_{0:T}})^2\cdot \mathrm{sup}_{\theta,\varphi,\chi}\sum_{t=1}^T\lambda_+\otimes\lambda_+\left(\left|h^{y_{0:T}}_{t,\theta,\varphi}\right|\right) \right\|_{\psi_{\alpha_*}}\leq BT\eqsp,
    $$
and for all $0\leq t\leq T$,
$$
 \||\log c_+(Y_t)|\vee|\log c_-(Y_t)|\|_{\psi_{\alpha_*}}\leq B\eqsp.
$$
\end{hypH}

\subsection{Oracle inequalities and consequences}
Our main result is an oracle inequality for the risk. The so-called variance term has the usual rate $1/n$ up to $\log n$ terms in the sample size $n$. It is proved to grow as much as $T^3$ in the length $T$ of the sample sequences.  
We assume that $\Theta$ and $\Phi$ are compact spaces, and that the sum of their diameters is bounded by $d_0$.
\begin{theorem}
    \label{th:oracle}
    Assume that H\ref{assum:strong:mixing}-H\ref{assum:pdata} hold. Then, there exist constants $c_0$, $c_1$, $c_2$, $\tilde{D}$ which depend on $\udup$, $\udlow$, $\alpha_*$, $A$, $B$ and $d_0$  only, such that  with probability at least $1 - c_0 \mathrm{exp}(-c_1\{d_*\log n\}^{1\wedge \alpha_*})$,
$$
\int m(\widehat\theta_{n,T},\widehat\varphi_{n,T},y_{0:T})p_{\mathcal{D}}(y_{0:T})\rmd\mu(y_{0:T})
\leq \mathrm{inf}_{\gamma>0}\left\{(1+\gamma)
\mathsf{E}_T 
+ c_2(1+\gamma^{-1})\frac{\tilde{D}d_* T^3}{n}\log(d_* n)(\log n)^{1/\alpha_*}\right\}\eqsp,    
$$
where  $\mathsf{E}_T = \mathrm{min}_{\theta\in\parspace,\varphi\in\parvarspace}
\int m(\theta,\varphi,y_{0:T})p_{\mathcal{D}}(y_{0:T})\rmd\mu(y_{0:T})
$ 
and $d_*=\dimt+\dimv$.
\end{theorem}
\begin{proof}
To prove Theorem \ref{th:oracle}, we use Theorem \ref{th:tang}, which is an alternative formulation of \cite[Theorem 3]{tang21a}, proved in Appendix \ref{app:tang}. First, Assumption A of Theorem \ref{th:oracle} holds with $D=\tilde{D}T$ for some positive constant $\tilde{D}$ depending on $B$.
This is a consequence of the first point in H\ref{assum:pdata}, Proposition~\ref{prop:assumA:likelihood} and Proposition~\ref{prop:assumA:kl}.

We now prove that Condition A of Theorem \ref{th:oracle} holds 
with 
$a_1 \leq C T^2$ for some $C>0$. Write, for all $\theta$, $\varphi_1$, $\varphi_2$, $y_{0:T}$, 
$$
\mathcal{E}^{y_{0:T}}(\theta,\varphi_1,\varphi_2) = \mathbb{E}_{q^{y_{0:T}}_{\varphi_1,0:T}}\left[\log \frac{q^{y_{0:T}}_{\varphi_2,0:T}(X_{0:T})}{\phi^{y_{0:T}}_{\theta,0:T|T}(X_{0:T})}\right]\eqsp.
$$
Note that
\begin{equation*}
\Delta(\theta, \theta', \varphi, \varphi', y_{0:T}) \leq \left|\ell_T^{y_{0:T}}(\theta) - \ell_T^{y_{0:T}}(\theta')\right| + 
\left|\mathcal{E}^{y_{0:T}}(\theta,\varphi,\varphi)- \mathcal{E}^{y_{0:T}}(\theta',\varphi',\varphi')\right|\,.
\end{equation*}
Write
\begin{equation*}
\left|\mathcal{E}^{y_{0:T}}(\theta,\varphi)- \mathcal{E}^{y_{0:T}}(\theta',\varphi')\right| \leq \Delta_1(\theta, \varphi, \varphi', y_{0:T}) 
+ \Delta_2(\theta, \theta', \varphi, \varphi', y_{0:T})\,,
\end{equation*}
where
\begin{align*}
\Delta_1(\theta, \varphi, \varphi', y_{0:T})  &= \left|\mathcal{E}^{y_{0:T}}(\theta,\varphi,\varphi) - \mathcal{E}^{y_{0:T}}(\theta,\varphi',\varphi)\right|\,,\\
\Delta_2(\theta, \theta', \varphi, \varphi', y_{0:T})&= \left|\mathcal{E}^{y_{0:T}}(\theta,\varphi',\varphi) - \mathcal{E}^{y_{0:T}}(\theta',\varphi',\varphi')\right|\eqsp.
\end{align*}
Therefore, 
\begin{equation*}
\Delta(\theta, \theta', \varphi, \varphi', y_{0:T}) \leq \left|\ell_T^{y_{0:T}}(\theta) - \ell_T^{y_{0:T}}(\theta')\right| + \Delta_1(\theta, \varphi, \varphi', y_{0:T}) + \Delta_2(\theta, \theta', \varphi, \varphi', y_{0:T})\eqsp.
\end{equation*}
By Proposition~\ref{prop:loglikelihood:lipschitz}, 
Proposition~\ref{prop:delta1} and Proposition~\ref{prop:delta2}, we get that
for all $\theta$, $\theta'$, $\varphi$, $\varphi'$, and all $y_{0:T}$,
$$
\Delta(\theta, \theta', \varphi, \varphi', y_{0:T}) \leq
\left(\kappa_{1}(y_{0:T})+\kappa_{4}(y_{0:T}) \right)\|\theta-\theta'\|_2 + 
\left(\kappa_{2}(y_{0:T})+\kappa_{3}(y_{0:T}) \right)\|\varphi-\varphi'\|_2,
$$
where
\begin{equation}
    \label{kappa1}
\kappa_1(y_{0:T}) = \frac{\udup \eta_+(G^{y_0}) }{\udlow c_-(y_0)} +\sum_{t=1}^T \frac{\udup}{\udlow c_-(y_t)}\left\{c_+(y_t)L_{t-1}(y_{0:t-1}) + \frac{\eta_+\otimes\mu(M \cdot \bar g^{y_{t-1}}\bar \otimes g^{y_t})}{\udlow c_-(y_{t-1})} + \eta_+(G^{y_t})\right\}\eqsp,
\end{equation}
with $M \cdot \bar g^{y_{t-1}} \otimes\bar g^{y_t}(x,x') = M(x,x')\bar g^{y_{t-1}}(x)\bar g^{y_t}(x')$, and for all $t$,
\begin{equation}
    \label{Lt}
L_t(y_{0:t}) =\frac{4\udup^2}{\udlow^2}\sum_{s=0}^t\varepsilon^{t-s} \frac{1}{c_-(y_s)}\left\{\frac{1}{\udlow c_-(y_{s-1})}\eta_+\otimes\mu\left(M \cdot\bar g^{y_{s-1}}\otimes\bar g^{y_{s}}\right) +  \mu(G^{y_s})\right\}\eqsp,
\end{equation}
with $\varepsilon = 1 - \udlow/\udup$,
\begin{equation}
    \label{kappa2}
    \kappa_2(y_{0:T}) = (\udupvar^{y_{0:T}})^3\sum_{t=1}^T\upsilon_t^{y_{0:T}}\sum_{s=t-1}^T \lambda_+\otimes \lambda_+   (K^{y_{0:T}}_{s|s+1}) \rho(y_{0:T})^{s-t}\eqsp,
\end{equation}
\begin{equation}
    \label{kappa3}
    \kappa_3(y_{0:T}) = \udupvar^{y_{0:T}}\left(\udupvar^{y_{0:T}}\sum_{t=1}^T c_{1,t}^{y_{0:T}} + c_{3,T}^{y_{0:T}} \right)\eqsp,
\end{equation}
and
\begin{equation}
    \label{kappa4}
    \kappa_4(y_{0:T}) = \udupvar^{y_{0:T}}\left(\udupvar^{y_{0:T}}\sum_{t=1}^T c_{2,t}^{y_{0:t-1}} + c_{4,t}^{y_{0:T}}\right)\eqsp,
\end{equation}
in which $\upsilon_t^{y_{0:T}}$, $c_{1,t}^{y_{0:T}}$, $c_{2,t}^{y_{0:t-1}}$, $c_{3,T}^{y_{0:T}}$ and $c_{4,t}^{y_{0:T}}$ are defined in H\ref{assum:boundH}.
Using H\ref{assum:moments}, it is easy to prove that 
$\E[\kappa_1(y_{0:T})^2]$, $\E[\kappa_2(y_{0:T})^2]$, $\E[\kappa_3(y_{0:T})^2]$, and $\E[\kappa_4(y_{0:T})^2]$ are upper bounded by $cT^2$ for  a constant $c$ that depends only on $\udup$, $\udlow$ and $A$, and  Theorem \ref{th:oracle} follows.
\end{proof}

Note that
$$
\int m(\widehat\theta_{n,T},\widehat\varphi_{n,T},y_{0:T})p_{\mathcal{D}}(y_{0:T})\rmd \mu(y_{0:T})
=\mathrm{KL}\left(P_\mathcal{D}\middle\|  P_{\widehat\theta_{n,T}}^Y\right) +  \E_{P_\mathcal{D}}\mathrm{KL}\left(Q^{Y^1_{0:T}}_{\widehat\varphi_{n,T},0:T}\middle\| \filtmeas_{\widehat\theta_{n,T},0:T|T}^{Y^1_{0:T}}\right).
$$
If the upper bound in Theorem \ref{th:oracle} is small, then the distribution $P_\mathcal{D}$ of the observations is well approximated by the decoding observational distribution $P_{\widehat\theta_{n,T}}^Y$, and the decoding distribution of the latent state distribution given data $\filtmeas_{\widehat\theta_{n,T},0:T|T}^{Y^1_{0:T}}$ is also in average well approximated by the coding distribution $Q^{Y^1_{0:T}}_{\widehat\varphi_{n,T},0:T}$.  Similarly,
$$\mathsf{E}_T = \mathrm{min}_{\theta\in\parspace,\varphi\in\parvarspace}
\left\{\mathrm{KL}\left(P_\mathcal{D}\middle\|  P_{\theta}^Y\right) +  \E_{P_\mathcal{D}}\mathrm{KL}\left(Q^{Y^1_{0:T}}_{\varphi,0:T}\middle\| \filtmeas_{\theta,0:T|T}^{Y^1_{0:T}}\right)\right\}.
$$
In case the data follows a state space distribution given by some decoding distribution, 
that is if there exists $\theta^*\in\Theta$ such that $P_\mathcal{D}=P_{\theta^*}^Y$, the oracle inequality in Theorem \ref{th:oracle} becomes, by taking $\theta=\theta^*$ to upper bound $\mathsf{E}_T$,
\begin{multline}
\mathrm{KL}\left(P_{\theta^*}^Y\middle\|  P_{\widehat\theta_{n,T}}^Y\right) +  \E_{P_{\theta^*}^Y}\mathrm{KL}\left(Q^{Y^1_{0:T}}_{\widehat\varphi_{n,T},0:T}\middle\| \filtmeas_{\widehat\theta_{n,T},0:T|T}^{Y^1_{0:T}}\right)
\leq (1+\gamma)
\mathrm{min}_{\varphi\in\parvarspace}\E_{P_{\theta^*}^Y}\mathrm{KL}\left(Q^{Y^1_{0:T}}_{\varphi,0:T}\middle\| \filtmeas_{\theta^*,0:T|T}^{Y^1_{0:T}}\right)
 \\
+ c_2(1+\gamma^{-1})\frac{Dd_* T^3}{n}\log(d_* n)(\log n)^{1/\alpha_*}\label{eq:in}  
\end{multline}
for any $\gamma>0$.
In the following corollary, we assume that the coding backward kernels are chosen such that they are good approximations of the backward decoding kernels in Kullback-Leibler divergence. 
\begin{hypH}
\label{assum:approx}
There exists $\epsilon >0$, such that for all $\theta\in\Theta$ there exists $\varphi\in\Phi$ such that for all $y_{0:T}\in\Yset^{T+1}$, $$\mathrm{KL}\left(Q^{y_{0:T}}_{\varphi,T}\middle\| \filtmeas_{\theta^*,T}^{y_{0:T}}\right)\leq \epsilon$$ and for all $1\leq t\leq T$, 
$$\mathrm{KL}\left(Q^{y_{0:T}}_{\varphi,t-1|t}\middle\| B^{y_{0:t-1}}_{\theta,t-1|t}\right)\leq \epsilon
\eqsp.
$$
\end{hypH}
\begin{corollary}
    \label{cor:kl}
Assume there exists $\theta^*\in\Theta$ such that $P_\mathcal{D}=P_{\theta^*}^Y$. 
Assume moreover H\ref{assum:approx}.
Then under the same assumptions as in Theorem \ref{th:oracle},  for the constants 
$c_0$, $c_1$, $c_2$, $D$ in Theorem \ref{th:oracle}, with probability at least $1 - c_0 \mathrm{exp}(-c_1\{d_*\log n\}^{1\wedge \alpha_*})$, for any $\gamma>0$,
\begin{multline*}
\mathrm{KL}\left(P_{\theta^*}^Y\middle\|  P_{\widehat\theta_{n,T}}^Y\right) +  \E_{P_{\theta^*}^Y}\mathrm{KL}\left(Q^{Y^1_{0:T}}_{\widehat\varphi_{n,T},0:T}\middle\| \filtmeas_{\widehat\theta_{n,T},0:T|T}^{Y^1_{0:T}}\right)\\
\leq (1+\gamma)
 (T+1)\epsilon
+ c_2(1+\gamma^{-1})\frac{Dd_* T^3}{n}\log(d_* n)(\log n)^{1/\alpha_*}\eqsp.   
\end{multline*}
\end{corollary}
When the data distribution is given by a state space model, Corollary~\ref{cor:kl} provides an upper bound for the Kullback-Leibler divergence between the data distribution
and its estimator and between the variational posterior and the estimated state space posterior distributions. This result sheds additional light on the quality of variational reconstruction in state space models with respect to  \cite[Proposition 3]{mathisJMLR}. In \cite[Proposition 3]{mathisJMLR}, the authors provided upper bounds on the error between conditional expectations of state functionals under the true posterior distribution and under its variational approximation. In both settings, designing coding backward kernels that are good
approximations of the true backward decoding kernels is enough to obtain quantitative controls on the reconstruction error. 
\begin{proof}
The result follows from \eqref{eq:in}, H\ref{assum:approx} and the fact that for any $\theta\in\Theta$ and $\varphi\in\Phi$, for any $y_{0:T}$,
$$
\mathrm{KL}\left(Q^{y_{0:T}}_{\varphi,0:T}\middle\| \filtmeas_{\theta^*,0:T|T}^{y_{0:T}}\right)=\sum_{t=1}^{T} \mathrm{KL}\left(Q^{y_{0:T}}_{\varphi,t-1|t}\middle\| B^{y_{0:t-1}}_{\theta,t-1|t}\right)
+\mathrm{KL}\left(Q^{y_{0:T}}_{\varphi,T}\middle\| \filtmeas_{\theta^*,T}^{y_{0:T}}\right).
$$

\end{proof}

\subsection{Applications}
\label{sec:appli}

In this section, we consider generative models where the  transition kernels and emission distributions are Gaussian in various classical settings. We show that under weak assumptions on these models, some assumptions of our main results hold. Establishing that all assumptions are satisfied in general settings, i.e. without very specific assumptions on the architectures, is a more challenging problem. 

We prove in Appendix~\ref{sec:check:assum} that H\ref{assum:strong:mixing} holds in particular for compact state spaces. We also prove that the functions $h_{t,\theta,\varphi}^{y_{0:T}}$ are upper-bounded explicitly,  and that $\phi_{\theta,t}^{y_{0:t}}$ and $b_{\theta,t-1|t}^{y_{0:t-1}}$ are lower and upper-bounded explicitly. This allows to obtain explicit constants in H\ref{assum:boundH}. Providing additional comments on the assumptions requires  assumptions on the observation space or on the dependency of the variational distributions on the observations. When the observation space is compact we can also obtain a uniform control with respect to the observations of these upper bounds which is crucial to check H\ref{assum:moments} and H\ref{assum:pdata}.

\paragraph{Gaussian backward kernels with dense networks. } We consider a generative model where the  transition kernels and emission distributions are Gaussian and parameterized by dense networks. 
\begin{itemize}
    \item For all $x\in\Xset$, $x'\mapsto m_\theta(x,x')$ is the Gaussian probability density function with mean $\mu_\theta(x)$, and variance $\Sigma_\theta(x)$ where
     $(\mu_\theta(x), \Sigma_\theta(x)) = \mathsf{MLP}^\theta(x)$ with $\mathsf{MLP}^\theta$ a dense Multi-layer network with input $x$ and weights given by $\theta$. In this case, if the output layer of $\mathsf{MLP}^\theta$  is such that $\mu_\theta$ is bounded and $\underline \Sigma \leq \Sigma_\theta^{-1}(x) \leq \overline \Sigma$ (i.e. $\Sigma_\theta^{-1}(x) - \underline \Sigma$ and $\overline \Sigma - \Sigma_\theta^{-1}(x)$  are positive semi-definite matrices) for all $x\in\Xset$, then there exist  constants $\underline c$, $\overline c$ such that for all $x,x'\in\Xset$,
     $$
\underline c\exp\left(-\overline\lambda x^\top x \right) \leq m_\theta(x',x) \leq \overline c\exp\left(-\underline\lambda \alpha(x) \right)\eqsp,
     $$
     where $\underline \lambda$ is the smallest eigenvalue of $\underline \Sigma$ and $\overline \lambda$ is the largest eigenvalue of $\overline \Sigma$ and where
     $$
\alpha(x) = \frac{1}{2}\left((\|x\|-M)^2\mathds{1}_{\|x\|\geq M} + (\|x\|-m)^2\mathds{1}_{\|x\|\leq m} + (M-m)^2\mathds{1}_{m\leq \|x\|\leq M}\right)\eqsp,
     $$
     with $m = \inf_{x\in\Xset, \theta\in\Theta} \|\mu_\theta(x)\|$ and $M = \sup_{x\in\Xset, \theta\in\Theta} \|\mu_\theta(x)\|$. This implies that H\ref{assum:strong:mixing} holds. In order to check H\ref{assum:lip}, if we assume also that for all $x\in\Xset$, $\theta \mapsto \mu_\theta(x)$ and  $\theta \mapsto \Sigma^{-1}_\theta(x)$ are continuously differentiable and that $\Theta$ is compact then there exists $M$ such that for all $\theta, \theta'\in\Theta$ and $x,x'\in\Xset$,
$$
\left| m_{\theta}(x,x') - m_{\theta'}(x,x')\right| \leq  M(x,x')\|\theta-\theta'\|_2\eqsp.
$$
We can check H\ref{assum:boundH} for $\log b^{y_{0:t-1}}_{\theta,t-1|t}$, as other items can be verified following the same steps. Assuming that $b^{y_{0:t-1}}_{\theta,t-1|t}(x,\cdot)$ is a Gaussian probability density with mean $\mu^{y_{0:t-1}}_{\theta,t-1|t}(x)$ and variance $\Sigma^{y_{0:t-1}}_{\theta,t-1|t}(x)$. Under similar regularity assumptions  on the networks providing $\mu^{y_{0:t-1}}_{\theta,t-1|t}(x)$ and $\Sigma^{y_{0:t-1}}_{\theta,t-1|t}(x)$, when $\Theta$ is compact,  H\ref{assum:boundH} holds. 
    \item For all $1\leq t \leq T$, $x\in\Xset$, $x'\mapsto q^{y_{0:T}}_{\varphi, t-1|t}(x,x')$ is the Gaussian probability density function with mean $\mu^{y_{0:T}}_{\varphi,t-1|t}(x)$, and variance $\Sigma^{y_{0:T}}_{\varphi,t-1|t}(x)$ where
     $(\mu^{y_{0:T}}_{\varphi,t-1|t}(x), \Sigma^{y_{0:T}}_{\varphi,t-1|t}(x)) = \mathsf{MLP}^{y_{0:T},\varphi}_{t-1|t}(x)$ with $\mathsf{MLP}^{y_{0:T},\varphi}_{t-1|t}$ a dense Multi-layer network with input $x$ and weights depending on $\varphi$. In this case, is the output layer of $\mathsf{MLP}^{y_{0:T},\varphi}_{t-1|t}$  is such that $ \mu_{\varphi,t-1|t}^{y_{0:T}}$ is bounded and $\underline \Sigma_{t-1|t}^{y_{0:T}} \leq (\Sigma_{\varphi,t-1|t}^{y_{0:T}}(x))^{-1} \leq \overline \Sigma_{t-1|t}^{y_{0:T}}$ (i.e. $(\Sigma^{y_{0:T}}_{\varphi,t-1|t}(x))^{-1} - \underline \Sigma^{y_{0:T}}_{t-1|t}$ and  $\overline \Sigma^{y_{0:T}}_{t-1|t}-(\Sigma^{y_{0:T}}_{\varphi,t-1|t}(x))^{-1}$  are positive semi-definite matrices) for all $x\in\Xset$, then there exist constants $\underline c^{y_{0:T}}_{t-1|t}$, $\overline c^{y_{0:T}}_{t-1|t}$ such that for all $x,x'\in\Xset$,
     $$
\underline c^{y_{0:T}}_{t-1|t}\exp\left(-\overline\lambda^{y_{0:T}}_{t-1|t} x^\top x \right) \leq q^{y_{0:T}}_{\varphi, t-1|t}(x',x) \leq \overline c^{y_{0:T}}_{t-1|t}\exp\left(-\underline\lambda^{y_{0:T}}_{t-1|t} \beta(x) \right)\eqsp,
     $$
      where $\underline \lambda^{y_{0:T}}_{t-1|t}$ is the smallest eigenvalue of $\underline \Sigma^{y_{0:T}}_{t-1|t}$ and $\overline \lambda^{y_{0:T}}_{t-1|t}$ is the largest eigenvalue of $\overline \Sigma^{y_{0:T}}_{t-1|t}$ and where
     \begin{multline*}
\beta(x) = \frac{1}{2}\left((\|x\|-M^{y_{0:T}}_{t-1|t})^2\mathds{1}_{\|x\|\geq M^{y_{0:T}}_{t-1|t}} + (\|x\|-m^{y_{0:T}}_{t-1|t})^2\mathds{1}_{\|x\|\leq m^{y_{0:T}}_{t-1|t}}  \right. \\
\left. + (M^{y_{0:T}}_{t-1|t}-m^{y_{0:T}}_{t-1|t})^2\mathds{1}_{m^{y_{0:T}}_{t-1|t}\leq \|x\|\leq M^{y_{0:T}}_{t-1|t}}\right)\eqsp,
     \end{multline*}
     with $m^{y_{0:T}}_{t-1|t} = \inf_{x\in\Xset} \|\mu^{y_{0:T}}_{t-1|t}(x)\|$ and $M^{y_{0:T}}_{t-1|t} = \sup_{x\in\Xset} \|\mu^{y_{0:T}}_{t-1|t}(x)\|$.
     Similar assumptions can be used for $q^{y_{0:T}}_{\varphi, T}$ using dense neural networks with bounded output. Under similar regularity assumptions on $\mu^{y_{0:T}}_{\varphi,t-1|t}$, and $\Sigma^{y_{0:T}}_{\varphi,t-1|t}$ than for $\mu_{\theta}$, and variance $\Sigma_{\theta}$, we may prove that H\ref{assum:lip} holds when $\Phi$ is compact.
\end{itemize}
 \paragraph{Gaussian backward kernels with recurrent networks. }  A natural parameterization is also to  use a recurrent neural network which updates an internal state $(s_t)_{t \geq 0}$ from which the backward variational kernels and filtering density are built. For all $t\geq 0$, define $s_t = \mathsf{RNN}^\varphi(s_{t-1}, y_t)$ where $ \mathsf{RNN}^\varphi$ is a recurrent neural network, and let  $x'\mapsto q^{y_{0:T}}_{\varphi, t-1|t}(x,x')$ be the Gaussian probability density function with mean $\mu^{y_{0:T}}_{t-1|t}$, and variance $\Sigma^{y_{0:T}}_{t-1|t}$  where $(\mu_t, \Sigma_t) = \mathsf{MLP}^\varphi(s_t)$. If the network $\mathsf{MLP}^\varphi$ is bounded similarly as in the dense neural network case, then the backward variational kernels satisfy H\ref{assum:strong:mixing}. 

\paragraph{Functional autoregressive models. } The discussion on neural networks also indicates that the assumptions can be verified for some classical statistical models. Assume for instance that $\Xset = \mathbb{R}$ and that for all $\theta\in\Theta$, $x\in\Xset$, $x'\mapsto m_\theta(x,x')$  is the Gaussian probability density function with mean $f_\theta(x)$, and variance $\sigma^2_\theta(x)$. Then,  H\ref{assum:strong:mixing} holds for $m_\theta$ when $-\infty < \inf_{x\in\Xset, \theta\in\Theta} f_\theta(x) \leq \sup_{x\in\Xset, \theta\in\Theta} f_\theta(x) <\infty$ and $-\infty < \inf_{x\in\Xset, \theta\in\Theta} \sigma_\theta(x) \leq \sup_{x\in\Xset, \theta\in\Theta} \sigma_\theta(x) <\infty$.

\paragraph{Gaussian emission densities. } Assume that at each time $t\geq 0$, $Y_t = h_\theta(X_t) + \varepsilon_t$, where $\{\varepsilon_t\}_{t\geq 0}$ are independent Gaussian random variables. Assume also that $h_\theta(X_t) = \mathsf{MLP}^\theta(X_t)$ where $\mathsf{MLP}^\theta$ is a dense neural network with bounded output layer, then H\ref{assum:bound:likelihood} holds. Assume that for all $x\in\Xset$, $\theta \mapsto h_\theta(x)$ is continuously differentiable and that $\Theta$ is compact, for all  $y\in\Yset$, there exists $G^{y}$ such that for all $\theta, \theta'\in\Theta$ and $x\in\Xset$,
$$
\left| g_\theta^{y}(x) - g_{\theta'}^{y}(x)\right| \leq  G^{y}(x)\|\theta-\theta'\|_2\eqsp,
$$
which means that H\ref{assum:lip} holds for the emission distributions.

\section{Discussion}
\label{sec:discussion}

 In this paper, we used a backward decomposition of variational posterior distributions to propose the first theoretical results for variational autoencoders (VAE) applied to general state space models.  Under strong mixing assumptions on the state space model and on the variational distribution, we provide in particular an oracle inequality and an upper bound for the Kullback-Leibler divergence between the data distribution and its estimator.

 Although these results are the first theoretical guarantees for VAE in the context of state space models, we believe that this is the first step to solve challenging open problems. First, in order to cover a wider variety of applications, weakening the strong mixing assumptions, for instance using local Doeblin assumptions, would be very interesting although it is a challenge when analyzing the stability of smoothing distributions. Another research direction is to understand how our results can be extended in settings where the observations are processed online, i.e. in cases where the parameters are updated when new observations are received but never stored. To the best of our knowledge, online variational estimation has recently been explored with new methodologies but without any theoretical guarantees.

\subsubsection*{Acknowledgements}
The research work of \'Elisabeth Gassiat was supported by the Institut Universitaire de France (IUF) and the French National Research Agency (ANR) with the projects ASCAI ANR-21-CE23-0035-02 and ANR BACKUP ANR-23-CE40-0018-02.

\bibliography{variationalhmm}
\bibliographystyle{apalike}

\appendix
\section{An oracle inequality adapted from \cite{tang21a}}
\label{app:tang}
We  propose an alternative formulation of Theorem 3 in \cite{tang21a} in which we provide the precise behavior of the constant in the variance term. To avoid introducing too many new notations, we formulate the results of \cite{tang21a} choosing the observation to be $Y_{0:T}$, the latent variables to be $X_{0:T}$ in our setting.
 \paragraph{Condition A. } There exist $a_1>0$ and a function $b$ such that for all $\theta\in\Theta$, $\theta'\in\Theta$, $\varphi\in\Phi$, $\varphi'\in\Phi$, $y_{0:T}\in\Yset^{T+1}$,
$$
\left|m(\theta,\varphi,y_{0:T})  - m(\theta',\varphi',y_{0:T}) \right| 
\leq b(y_{0:T})\|(\theta,\varphi)-(\theta',\varphi')\|_2 \,,
$$
with $\mathbb{E}[b^2(Y_{0:T})]\leq a_1$.

 \paragraph{Assumption A. } There exist $\alpha_*>0$ and $D>0$ such that
\begin{equation}
\label{eq:assumptionA:app}
\left\|\mathrm{sup}_{\theta,\varphi}\left\{\left|\log \frac{L_T^{Y_{0:T}}(\theta)}{p_{\mathcal{D}}(Y_{0:T})}\right|  + \mathrm{KL}\left(Q^{Y_{0:T}}_{\varphi,0:T}\middle\| \phi_{\theta,T}^{Y_{0:T}}\right)\right\}\right\|_{\psi_{\alpha_*}}
\leq D\eqsp.
\end{equation}

\begin{theorem}
    \label{th:tang}
    Assume that $\Theta$ and $\Phi$ are compact spaces and that the sum of their diameter is upper bounded by $d_0$. Assume moreover that Condition A and Assumption A hold. Then, there exist constants $c_0$, $c_1$, which depend on $d_0$, $a_1$ and $\alpha_*$, and a universal constant $c_2$, such that  with probability at least $1 - c_0 \mathrm{exp}(-c_1\{d_*\log n\}^{1\wedge \alpha_*})$,
$$
\int m(\widehat\theta_{n,T},\widehat\varphi_{n,T},y_{0:T})p_{\mathcal{D}}(y_{0:T})\rmd\mu(y_{0:T})
\leq \mathrm{inf}_{\gamma>0}\left\{(1+\gamma)
\mathsf{E}_T 
+ c_2(1+\gamma^{-1})\frac{a_{1} Dd_*}{n}\log(d_* n)(\log n)^{1/\alpha_*}\right\}\eqsp,    
$$
where  $\mathsf{E}_T = \mathrm{min}_{\theta\in\parspace,\varphi\in\parvarspace}
\int m(\theta,\varphi,y_{0:T})p_{\mathcal{D}}(y_{0:T})\rmd\mu(y_{0:T})
$ 
and $d_*=\dimt+\dimv$.
\end{theorem}

\begin{proof}
 We follow 
 the proof of \cite[Theorem 3]{tang21a}, in which we track the dependencies of the constants with respect to $a_1$. 
 In \cite[Lemma 14]{tang21a}, a multiplicative term $\sqrt{a_1}$ is required on the r.h.s. of the inequality. Then on page 24 third line the inequality needs again $\sqrt{a_1}$ on the r.h.s., and the end of the proof follows by multiplying $\delta_n$ by $\sqrt{a_1}$. We obtain that in  \cite[Theorem 3]{tang21a}, their constant $c_2$ is proportional to $a_1$. 
    \end{proof}

\section{Additional proofs}
\label{app:A}

\begin{proposition}
\label{prop:loglikelihood:lipschitz}
Assume that H\ref{assum:strong:mixing}-\ref{assum:lip} hold. For all $\theta$, $\theta'\in\Theta$, and all $y_{0:T}\in\Yset^{T+1}$,
$$
\left|\ell_T^{y_{0:T}}(\theta) - \ell_T^{y_{0:T}}(\theta')\right| \leq \kappa_1(y_{0:T})\|\theta-\theta'\|_2\eqsp,
$$
where 
$$
\kappa_1(y_{0:T}) = \frac{\udup \eta_+(G^{y_0}) }{\udlow c_-(y_0)} +\sum_{t=1}^T \frac{\udup}{\udlow c_-(y_t)}\left\{c_+(y_t)L_{t-1}(y_{0:t-1}) + \frac{\eta_+\otimes\mu(M \cdot \bar g^{y_{t-1}}\bar \otimes g^{y_t})}{\udlow c_-(y_{t-1})} + \eta_+(G^{y_t})\right\}\eqsp,
$$
with $M \cdot \bar g^{y_{t-1}} \otimes\bar g^{y_t}(x,x') = M(x,x')\bar g^{y_{t-1}}(x)\bar g^{y_t}(x')$, where $L_{t-1}$ is defined in Lemma~\ref{lem:filtering:lip}. 
\end{proposition}
\begin{proof}
For all $\theta$, $\theta'\in\Theta$, and all $y_{0:T}\in\Yset^{T+1}$, with the convention $p_\theta(y_0|y_{-1}) = p_\theta(y_0)$,
$$
\ell_T^{y_{0:T}}(\theta) - \ell_T^{y_{0:T}}(\theta') = \sum_{t=0}^T \left(\log p_\theta(y_t|y_{0:t-1})-\log p_{\theta'}(y_t|y_{0:t-1})\right)\eqsp.
$$
For all $t>0$, 
$$
p_\theta(y_t|y_{0:t-1}) = \int \filtmeas^{y_{0:t-1}}_{\theta,t-1}(\rmd x_{t-1})M_\theta(x_{t-1},\rmd x_t)g_\theta^{y_t}(x_t)\eqsp.
$$
Note first that
$$
p_\theta(y_t|y_{0:t-1}) \geq \udlow c_-(y_t)\eqsp,
$$
so that
$$
\left|\ell_T^{y_{0:T}}(\theta) - \ell_T^{y_{0:T}}(\theta')\right| \leq \frac{\left| p_\theta(y_{0})- p_{\theta'}(y_0)\right|}{\udlow c_-(y_0)}+
\sum_{t=0}^T \frac{\left|p_\theta(y_t|y_{0:t-1})-p_{\theta'}(y_t|y_{0:t-1})\right|}{\udlow c_-(y_t)}\eqsp.
$$
For $t=0$, using that $p_\theta(y_{0}) = \int \chi(\rmd x_0)g^{y_0}_\theta(x_0)$,  Assumptions  H\ref{assum:strong:mixing} and H\ref{assum:lip} yield
$$
\left| p_\theta(y_{0})- p_{\theta'}(y_0)\right|\leq \udup \eta_+(G^{y_0}) \|\theta-\theta'\|_2\eqsp.
$$
In addition, 
\begin{multline*}
p_\theta(y_t|y_{0:t-1})-p_{\theta'}(y_t|y_{0:t-1})= \int \left(\filtmeas^{y_{0:t-1}}_{\theta,t-1}(\rmd x_{t-1})-\filtmeas^{y_{0:t-1}}_{\theta',t-1}(\rmd x_{t-1})\right)M_\theta(x_{t-1},\rmd x_t)g_\theta^{y_t}(x_t)\\
+
\int \filtmeas^{y_{0:t-1}}_{\theta',t-1}(\rmd x_{t-1})\left(M_\theta(x_{t-1},\rmd x_t)-M_{\theta'}(x_{t-1},\rmd x_t)\right)g_\theta^{y_t}(x_t)
+\int \filtmeas^{y_{0:t-1}}_{\theta',t-1}(\rmd x_{t-1})M_{\theta'}(x_{t-1},\rmd x_t)\left(g_\theta^{y_t}(x_t)-g_{\theta'}^{y_t}(x_t)\right)
\eqsp.
\end{multline*}
Using Lemma~\ref{lem:bound:filter}, Assumptions  H\ref{assum:strong:mixing} and H\ref{assum:lip}, we get
\begin{multline*}
\left|p_\theta(y_t|y_{0:t-1})-p_{\theta'}(y_t|y_{0:t-1})\right|\leq \left\{\sigma_+c_+(y_t)\left\|\filtmeas^{y_{0:t-1}}_{\theta,t-1}-\filtmeas^{y_{0:t-1}}_{\theta',t-1}\right\|_{\mathrm{tv}} + \right.\\
\left. \frac{\udup}{\udlow c_-(y_{t-1})}\int \eta_+\otimes\mu(\rmd x \rmd x')(M(x,x') \bar g^{y_{t-1}}(x)
\bar g^{y_t}(x')) + \udup \eta_+(G^{y_t})\right\}\|\theta-\theta'\|_2\eqsp.
\end{multline*}

The proof is completed by using  Lemma~\ref{lem:filtering:lip}.
\end{proof}

\begin{proposition}
\label{prop:delta1}
Assume that H\ref{assum:strong:mixing}-\ref{assum:boundH} hold. Then,
$$
\Delta_1(\theta, \varphi, \varphi', y_{0:T}) \leq \kappa_2(y_{0:T})\|\varphi-\varphi'\|_2 \eqsp,
$$
where 
$$
\Delta_1(\theta, \varphi, \varphi', y_{0:T}) = \left|\mathbb{E}_{q^{y_{0:T}}_{\varphi,0:T}}\left[\log \frac{q^{y_{0:T}}_{\varphi,0:T}(X_{0:T})}{\phi^{y_{0:T}}_{\theta,0:T|T}(X_{0:T})}\right] - \mathbb{E}_{q^{y_{0:T}}_{\varphi',0:T}}\left[\log \frac{q^{y_{0:T}}_{\varphi,0:T}(X_{0:T})}{\phi^{y_{0:T}}_{\theta,0:T|T}(X_{0:T})}\right]\right|\eqsp,
$$
with $\rho(y_{0:T}) = 1-\udlowvar^{y_{0:T}}$ and
$$
\kappa_2(y_{0:T}) = (\udupvar^{y_{0:T}})^3\sum_{t=1}^T\upsilon_t^{y_{0:T}}\sum_{s=t-1}^T \lambda_+\otimes \lambda_+   (K^{y_{0:T}}_{s|s+1}) \rho(y_{0:T})^{s-t}\eqsp.
$$
\end{proposition}
\begin{proof}
For all $\varphi,\varphi'\in\Phi$, $0\leq t \leq T-1$, define
\begin{multline*}
\tilde q^{y_{0:T}}_{\varphi,\varphi',t|T}(x_{0:T}) = q^{y_{0:T}}_{\varphi,T}(x_T)\prod_{u=T}^{t+1}q^{y_{0:T}}_{\varphi,u-1|u}(x_u,x_{u-1})\prod_{u=t}^{1}q^{y_{0:T}}_{\varphi',u-1|u}(x_u,x_{u-1}) \\
- q^{y_{0:T}}_{\varphi,T}(x_T)\prod_{u=T}^{t+2}q^{y_{0:T}}_{\varphi,u-1|u}(x_u,x_{u-1})\prod_{u=t+1}^{1}q^{y_{0:T}}_{\varphi',u-1|u}(x_u,x_{u-1})
\end{multline*}
with the convention $\prod_{u=T}^{T+1}q^{y_{0:T}}_{\varphi,u-1|u}(x_u,x_{u-1}) = 1$ and $\prod_{u=0}^{1}q^{y_{0:T}}_{\varphi',u-1|u}(x_u,x_{u-1}) = 1$, and for $t=T$,
$$
\tilde q^{y_{0:T}}_{\varphi,\varphi',T|T}(x_{0:T}) = q^{y_{0:T}}_{\varphi,T}(x_T)\prod_{u=T}^1q^{y_{0:T}}_{\varphi',u-1|u}(x_u,x_{u-1}) 
- q^{y_{0:T}}_{\varphi',T}(x_T)\prod_{u=T}^1q^{y_{0:T}}_{\varphi',u-1|u}(x_u,x_{u-1})\eqsp.
$$
Therefore,
\begin{align*}
\Delta_1(\theta, \varphi, \varphi', y_{0:T})&= \left|\sum_{t=1}^T\mathbb{E}_{q^{y_{0:T}}_{\varphi,0:T}}\left[h^{y_{0:T}}_{t,\theta,\varphi}(X_{t-1},X_t)\right] - \mathbb{E}_{q^{y_{0:T}}_{\varphi',0:T}}\left[h^{y_{0:T}}_{t,\theta,\varphi}(X_{t-1},X_t)\right]\right|\eqsp,\\
&= \left|\sum_{t=1}^T\sum_{s=0}^T \mathbb{E}_{\tilde q^{y_{0:T}}_{\varphi,\varphi',s|T}}\left[h^{y_{0:T}}_{t,\theta,\varphi}(X_{t-1},X_t)\right]\right| \eqsp,
\end{align*}
where $h^{y_{0:T}}_{t,\theta,\varphi}$, $1\leq t\leq T$, are defined in \eqref{eq:def:addfunc}. 
Note first that if $t> s+1$, then $\mathbb{E}_{\tilde q^{y_{0:T}}_{\varphi,\varphi',s|T}}\left[h^{y_{0:T}}_{t,\theta,\varphi}(X_{t-1},X_t)\right] = 0$ so that
$$
\Delta_1(\theta, \varphi, \varphi', y_{0:T})= \left|\sum_{t=1}^T \sum_{s=t-1}^T \mathbb{E}_{\tilde q^{y_{0:T}}_{\varphi,\varphi',s|T}}\left[h^{y_{0:T}}_{t,\theta,\varphi}(X_{t-1},X_t)\right]\right| \eqsp.
$$
For all $t\leq s+1$, write for all measurable set $A$,
\begin{align*}
\mu^{y_{0:T}}_{\varphi,s}(A) &= \int \mathds{1}_A(x_s) q^{y_{0:T}}_{\varphi,T}(x_T)\mu(\rmd x_T)\prod_{u=T}^{s+1}q^{y_{0:T}}_{\varphi,u-1|u}(x_u,x_{u-1})\mu(\rmd x_{u-1})\eqsp,\\
\tilde\mu^{y_{0:T}}_{\varphi,\varphi',s}(A) &= \int \mathds{1}_A(x_s) q^{y_{0:T}}_{\varphi,T}(x_T)\mu(\rmd x_T)\prod_{u=T}^{s+2}q^{y_{0:T}}_{\varphi,u-1|u}(x_u,x_{u-1})\mu(\rmd x_{u-1})q^{y_{0:T}}_{\varphi',s|s+1}(x_{s+1},x_{s})\mu(\rmd x_{s})\eqsp.
\end{align*}
Therefore,
$$
\mathbb{E}_{\tilde q^{y_{0:T}}_{\varphi,\varphi',s|T}}\left[h^{y_{0:T}}_{t,\theta,\varphi}(X_{t-1},X_t)\right] = \left(\mu^{y_{0:T}}_{\varphi,s} - \tilde\mu^{y_{0:T}}_{\varphi,\varphi',s}\right)\left\{\prod_{u=s}^{t+1}Q^{y_{0:T}}_{\varphi',u-1|u}\right\}Q^{y_{0:T}}_{\varphi',t-1|t}h^{y_{0:T}}_{t,\theta,\varphi}\eqsp.
$$
Using H\ref{assum:strong:mixing}, the backward variational kernels satisfy a Doeblin condition, see \cite[Section~6.1.3]{douc2014nonlinear}, so that 
$$
\mathbb{E}_{\tilde q^{y_{0:T}}_{\varphi,\varphi',s|T}}\left[h^{y_{0:T}}_{t,\theta,\varphi}(X_{t-1},X_t)\right]  \leq \frac{1}{2}\|\mu^{y_{0:T}}_{\varphi,s}-\tilde\mu^{y_{0:T}}_{\varphi,\varphi',s}\|_{\mathrm{tv}}\rho(y_{0:T})^{s-t}\mathrm{osc}\left(Q^{y_{0:T}}_{\varphi',t-1|t}h^{y_{0:T}}_{t,\theta,\varphi}\right)\eqsp,
$$
where for all measurable functions $f$, $\mathrm{osc}(f) = \sup_{x,x'\in\Xset}|f(x) - f(x')|$. By H\ref{assum:strong:mixing} and H\ref{assum:boundH},
\begin{align*}
\mathrm{osc}\left(Q^{y_{0:T}}_{\varphi',t-1|t}h^{y_{0:T}}_{t,\theta,\varphi}\right) &\leq 2 \left\| \int q^{y_{0:T}}_{\varphi',t-1|t}(\cdot,x_{t-1})h^{y_{0:T}}_{t,\theta,\varphi}(x_{t-1},\cdot)\mu(\rmd x_{t-1})\right\|_\infty\eqsp,\\
&\leq 2 \udupvar^{y_{0:T}} \left\| \int \left|h^{y_{0:T}}_{t,\theta,\varphi}(x_{t-1},\cdot)\right|\lambda_+(\rmd x_{t-1})\right\|_\infty\eqsp,\\
&\leq 2 \udupvar^{y_{0:T}}\upsilon_t^{y_{0:T}}\eqsp.
\end{align*}
Noting that by H\ref{assum:lip},
$$
\|\mu^{y_{0:T}}_{\varphi,s}-\tilde\mu^{y_{0:T}}_{\varphi,\varphi',s}\|_{\mathrm{tv}} \leq Q^{y_{0:T}}_{\varphi,T}\prod_{s=T}^{t+1}Q^{y_{0:T}}_{\varphi,s-1|s}K^{y_{0:T}}_{s|s+1}\|\varphi-\varphi'\|_2\leq (\udupvar^{y_{0:T}})^2\lambda_+\otimes \lambda_+ (K^{y_{0:T}}_{s|s+1})\|\varphi-\varphi'\|_2\eqsp,
$$
concludes the proof.
\end{proof}

\begin{proposition}
\label{prop:delta2}
Assume that H\ref{assum:strong:mixing}-\ref{assum:boundH} hold. Then,
$$
\Delta_2(\theta, \theta', \varphi, \varphi', y_{0:T})
\leq \kappa_3(y_{0:T}) \left\|\varphi-\varphi'\right\|_2 + \kappa_4(y_{0:T}) \left\|\theta-\theta'\right\|_2\eqsp,
$$
where
$$
\Delta_2(\theta, \theta', \varphi, \varphi', y_{0:T}) = \left|\mathbb{E}_{q^{y_{0:T}}_{\varphi',0:T}}\left[\log \frac{q^{y_{0:T}}_{\varphi,0:T}(X_{0:T})}{\phi^{y_{0:T}}_{\theta,0:T|T}(X_{0:T})}\right] - \mathbb{E}_{q^{y_{0:T}}_{\varphi',0:T}}\left[\log \frac{q^{y_{0:T}}_{\varphi',0:T}(X_{0:T})}{\phi^{y_{0:T}}_{\theta',0:T|T}(X_{0:T})}\right]\right|\eqsp,
$$
with 
$$
\kappa_3(y_{0:T}) = \udupvar^{y_{0:T}}\left(\udupvar^{y_{0:T}}\sum_{t=1}^T c_{1,t}^{y_{0:T}} + c_{3,T}^{y_{0:T}} \right)\quad \mathrm{and} \quad
\kappa_4(y_{0:T}) = \udupvar^{y_{0:T}}\left(\udupvar^{y_{0:T}}\sum_{t=1}^T c_{2,t}^{y_{0:t-1}} + c_{4,t}^{y_{0:T}}\right)\eqsp,
$$
and where $c_{1,t}^{y_{0:T}}$, $c_{2,t}^{y_{0:t-1}}$, $c_{3,T}^{y_{0:T}}$ and $c_{4,t}^{y_{0:T}}$ are defined in H\ref{assum:boundH}.
\end{proposition}
\begin{proof}
By definition, 
\begin{align*}
\Delta_2(\theta, \theta', \varphi, \varphi', y_{0:T}) &= \left|\mathbb{E}_{q^{y_{0:T}}_{\varphi',0:T}}\left[\log \frac{q^{y_{0:T}}_{\varphi,0:T}(X_{0:T})}{\phi^{y_{0:T}}_{\theta,0:T|T}(X_{0:T})}\right] - \mathbb{E}_{q^{y_{0:T}}_{\varphi',0:T}}\left[\log \frac{q^{y_{0:T}}_{\varphi',0:T}(X_{0:T})}{\phi^{y_{0:T}}_{\theta',0:T|T}(X_{0:T})}\right]\right|\,,\\
&\leq \mathbb{E}_{q^{y_{0:T}}_{\varphi',0:T}}\left[\left|\log \frac{q^{y_{0:T}}_{\varphi,0:T}(X_{0:T})}{\phi^{y_{0:T}}_{\theta,0:T|T}(X_{0:T})} - \log \frac{q^{y_{0:T}}_{\varphi',0:T}(X_{0:T})}{\phi^{y_{0:T}}_{\theta',0:T|T}(X_{0:T})}\right|\right]\,,\\
&\leq \sum_{t=1}^T\mathbb{E}_{q^{y_{0:T}}_{\varphi',0:T}}\left[\left|h^{y_{0:T}}_{t,\theta,\varphi}(X_{t-1},X_t)-h^{y_{0:T}}_{t,\theta',\varphi'}(X_{t-1},X_t)\right|\right]\,,
\end{align*}
where $h^{y_{0:T}}_{t,\theta,\varphi}$, $1\leq t\leq T$,  are defined in \eqref{eq:def:addfunc}. For $t<T$ and all $x_{t-1}$, $x_t\in\Xset$, 
\begin{multline*}
\left|h^{y_{0:T}}_{t,\theta,\varphi}(x_{t-1},x_t)-h^{y_{0:T}}_{t,\theta',\varphi'}(x_{t-1},x_t)\right| \leq \left|\log q^{y_{0:T}}_{\varphi,t-1|t}(x_t,x_{t-1}) - \log q^{y_{0:T}}_{\varphi',t-1|t}(x_t,x_{t-1})\right|\\
+ \left|\log b^{y_{0:t-1}}_{\theta,t-1|t}(x_t,x_{t-1})- \log b^{y_{0:t-1}}_{\theta',t-1|t}(x_t,x_{t-1})\right|\eqsp.
\end{multline*}
Using H\ref{assum:strong:mixing} and H\ref{assum:boundH},
\begin{align*}
&\mathbb{E}_{q^{y_{0:T}}_{\varphi',0:T}}\left[\left|\log q^{y_{0:T}}_{\varphi,t-1|t}(x_t,x_{t-1}) - \log q^{y_{0:T}}_{\varphi',t-1|t}(x_t,x_{t-1})\right|\right] \\
&\hspace{2cm}\leq (\udupvar^{y_{0:T}})^{2}\int \lambda_+\otimes \lambda_+ (\rmd x \rmd x' )\left|\log q^{y_{0:T}}_{\varphi,t-1|t}(x,x') - \log q^{y_{0:T}}_{\varphi',t-1|t}(x,x')\right|\eqsp,\\
&\hspace{2cm}\leq (\udupvar^{y_{0:T}})^{2}c_{1,t}^{y_{0:T}}\left\|\varphi-\varphi'\right\|_2\eqsp.
\end{align*}
Similarly,
\begin{align*}
&\mathbb{E}_{q^{y_{0:T}}_{\varphi',0:T}}\left[\left|\log b^{y_{0:t-1}}_{\theta,t-1|t}(x_t,x_{t-1}) - \log b^{y_{0:t-1}}_{\theta',t-1|t}(x_t,x_{t-1})\right|\right] \\
&\hspace{2cm}\leq (\udupvar^{y_{0:T}})^{2}\int \lambda_+\otimes \lambda_+ (\rmd x \rmd x' )\left|\log b^{y_{0:t-1}}_{\theta,t-1|t}(x,x') - \log b^{y_{0:t-1}}_{\theta',t-1|t}(x,x')\right|\eqsp,\\
&\hspace{2cm}\leq (\udupvar^{y_{0:T}})^{2}c_{2,t}^{y_{0:t-1}}\left\|\theta-\theta'\right\|_2\eqsp.
\end{align*}
For $t=T$, it remains to bound $\mathbb{E}_{q^{y_{0:T}}_{\varphi',0:T}}[|\log q_{\varphi,T}^{y_{0:T}}(X_T) - \log q_{\varphi',T}^{y_{0:T}}(X_T)| + |\log \phi_{\theta,T}^{y_{0:T}}(X_T) - \log \phi_{\theta',T}^{y_{0:T}}(X_T)|]$, which is straightforward by using H\ref{assum:strong:mixing} and H\ref{assum:boundH}.
\end{proof}

\begin{proposition}
    \label{prop:assumA:likelihood}
    Assume that H\ref{assum:strong:mixing}-\ref{assum:bound:likelihood} and H\ref{assum:pdata} hold. Then, there exists $c>0$ such that ,
    $$\left\|\sup_{\theta\in\Theta}\left|\log L_T^{Y_{0:T}}(\theta)\right|\right\|_{\psi_{\alpha_*}}\leq cT\eqsp.
    $$
\end{proposition}
\begin{proof}
    For all $\theta\in\Theta$, and all $y_{0:T}\in\Yset^{T+1}$, with the convention $p_\theta(y_0|y_{-1}) = p_\theta(y_0)$,
$$
\log L_T^{y_{0:T}}(\theta) = \ell_T^{y_{0:T}}(\theta) = \sum_{t=0}^T \log p_\theta(y_t|y_{0:t-1})\eqsp.
$$
As $p_\theta(y_0) = \int \chi(\rmd x_0)g_\theta^{y_0}(x_0) $, by H\ref{assum:strong:mixing}-\ref{assum:bound:likelihood}, $\udlow c_-(y_0) \leq p_\theta(y_0) \leq \udup c_+(y_0)$. For all $t>0$, 
$$
p_\theta(y_t|y_{0:t-1}) = \int \filtmeas^{y_{0:t-1}}_{\theta,t-1}(\rmd x_{t-1})M_\theta(x_{t-1},\rmd x_t)g_\theta^{y_t}(x_t)\eqsp,
$$
so  that by H\ref{assum:strong:mixing}-\ref{assum:bound:likelihood} $\udlow c_-(y_t) \leq p_\theta(y_t|y_{0:t-1}) \leq \udup c_+(y_t)$. Using  the second point in H\ref{assum:pdata} and the triangular inequality  concludes the proof.
\end{proof}

\begin{proposition}
    \label{prop:assumA:kl}
    Assume that H\ref{assum:strong:mixing} and H\ref{assum:pdata} hold. Then, there exists $B>0$ such that
    $$
\left\|\mathrm{sup}_{\theta\in\Theta,\varphi\in\Phi,\chi}\left| \mathrm{KL}\left(Q^{Y_{0:T}}_{\varphi,0:T}\middle\| \phi_{\theta,T}^{Y_{0:T}}\right)\right|\right\|_{\psi_{\alpha_*}}\leq B T\eqsp,
$$
\end{proposition}
\begin{proof}
For all $\theta\in\Theta$, $\varphi\in\Phi$, $y_{0:T}\in\Yset^{T+1}$,
$$
\mathrm{KL}\left(Q^{y_{0:T}}_{\varphi,0:T}\middle\| \phi_{\theta,T}^{y_{0:T}}\right) = \mathbb{E}_{q^{y_{0:T}}_{\varphi,0:T}}\left[\log \frac{q^{y_{0:T}}_{\varphi,0:T}(X_{0:T})}{\phi^{y_{0:T}}_{\theta,0:T|T}(X_{0:T})}\right] = \sum_{t=1}^T\mathbb{E}_{q^{y_{0:T}}_{\varphi,0:T}} \left[h^{y_{0:T}}_{t,\theta,\varphi}(X_{t-1},X_t)\right]\eqsp,
$$
where $h^{y_{0:T}}_{t,\theta,\varphi}$, $1\leq t\leq T$,  are defined in \eqref{eq:def:addfunc}. By H\ref{assum:strong:mixing}, for all $1\leq t \leq T$,
$$
\left|\mathbb{E}_{q^{y_{0:T}}_{\varphi,0:T}} \left[h^{y_{0:T}}_{t,\theta,\varphi}(X_{t-1},X_t)\right] \right| \leq (\udupvar^{y_{0:T}})^2\lambda_+\otimes\lambda_+\left(\left|h^{y_{0:T}}_{t,\theta,\varphi}\right|\right)\eqsp,
$$
which concludes the proof by H\ref{assum:pdata}.
\end{proof}


\section{Technical results}
\label{app:B}

\begin{lemma}
\label{lem:bound:filter}
Assume that H\ref{assum:strong:mixing} and H\ref{assum:bound:likelihood} hold. 
For all $\theta\in\Theta$,  all $t\geq 0$, all $y_{0:t}\in\Yset^{T+1}$, positive measurable function $h$, 
$$
\frac{\sigma_- \eta_-(g^{y_t}_\theta h)}{\sigma_+c_+(y_t)}\leq \filtmeas^{y_{0:t}}_{\theta,t}(h)\leq \frac{\sigma_+ \eta_+(g^{y_t}_\theta h)}{\sigma_-c_-(y_t)}\,.
$$
\end{lemma}
\begin{proof}
At time 0, we have $\filtmeas^{y_{0}}_{\theta,0}(\rmd x_{0}) \propto \chi(\rmd x_0)g^{y_0}_\theta(x_0)$ so that by H\ref{assum:strong:mixing}-\ref{assum:bound:likelihood},
$$
\frac{\sigma_- \eta_-(g^{y_0}_\theta h)}{\sigma_+c_+(y_0)}\leq \filtmeas^{y_{0}}_{\theta,0}(h)\leq \frac{\sigma_+ \eta_+(g^{y_0}_\theta h)}{\sigma_-c_-(y_0)}\eqsp.
$$
Similarly, 
$$
\filtmeas^{y_{0:t}}_{\theta,t}(\rmd x_{t}) \propto g^{y_t}_\theta(x_t)\int \filtmeas^{y_{0:t-1}}_{\theta,t-1}(\rmd x_{t-1})M_\theta(x_{t-1},\rmd x_t)\,,
$$
so that by H\ref{assum:strong:mixing} and H\ref{assum:bound:likelihood},
$$
\frac{\sigma_- \eta_- (g^{y_t}_\theta h)}{\sigma_+c_+(y_t)}\leq \filtmeas^{y_{0:t}}_{\theta,t}(h)\leq \frac{\sigma_+ \eta_+ (g^{y_t}_\theta h)}{\sigma_-c_-(y_t)}\,.
$$
\end{proof}

\begin{lemma}
\label{lem:filtering:lip}
Assume that H\ref{assum:strong:mixing}, H\ref{assum:bound:likelihood} and H\ref{assum:lip} hold. Then, for all $\theta$, $\theta'\in\Theta$, $t\geq 1$,
$$
\left\|\filtmeas_{\theta,t}^{y_{0:t}} - \filtmeas_{\theta',t}^{y_{0:t}}\right\|_{\mathrm{tv}} \leq L_t(y_{0:t})\|\parvec-\parvec'\|_2\eqsp,
$$
where
$$
L_t(y_{0:t}) =\frac{4\udup^2}{\udlow^2}\sum_{s=0}^t\varepsilon^{t-s} \frac{1}{c_-(y_s)}\left\{\frac{1}{\udlow c_-(y_{s-1})}\eta_+\otimes\mu\left(\bar g^{y_{s-1}}\otimes\bar g^{y_{s}}\cdot M\right) +  \eta_+(G^{y_s})\right\}\eqsp,
$$
with $\varepsilon = 1 - \udlow/\udup$.

\end{lemma}
\begin{proof}

The proof follows the same lines as the proof of \cite[Proposition~2.1]{de2017consistent}, which was in the setting of a discrete state space. For $t>0$, note that $\filtmeas_{\theta,t}^{y_{0:t}}(\rmd x_{t}) =  g_\theta^{y_t}(x_t)\int \filtmeas_{\theta,t-1}^{y_{0:t-1}}(\rmd x_{t-1})M_\theta(x_{t-1},\rmd x_t)/c_{\theta,t}(y_{0:t})$ where $c_{\theta,t}(y_{0:t}) = \int g_\theta^{y_t}(x_t) \filtmeas_{\theta,t-1}^{y_{0:t-1}}(\rmd x_{t-1})M_\theta(x_{t-1},\rmd x_t)$. Consider the forward kernel at time $t$ defined, for all $\parvec\in\parspace$, all $y_t\in\Yset$, $x\in\mathbb{R}^d$, and probability measure $\gamma$ by
$$
F_{\parvec,t}^{y_t}\gamma(x) = \frac{\int m_\parvec(x',x)g_\parvec^{y_t}(x)\gamma(\rmd x')}{\int m_\parvec(x',x'')g_\parvec^{y_t}(x'')\gamma(\rmd x')\mu(\rmd x'')}\eqsp.
$$
Therefore, $\filtmeas_{\theta,t}^{y_{0:t}} = F_{\parvec,t}^{y_t}\filtmeas_{\theta,t-1}^{y_{0:t-1}}$ and for all $\parvec, \parvec'\in\parspace$,
\begin{align*}
\filtmeas_{\theta,t}^{y_{0:t}} - \filtmeas_{\theta',t}^{y_{0:t}} &= F_{\parvec,t}^{y_t}\filtmeas_{\theta,t-1}^{y_{0:t-1}} - F_{\parvec',t}^{y_t}\filtmeas_{\theta',t-1}^{y_{0:t-1}}\eqsp,\\
&= \sum_{s=0}^{t-1}\Delta_{t,s}(y_{s:t}) + F_{\parvec,t}^{y_t}\filtmeas_{\theta',t-1}^{y_{0:t-1}} - F_{\parvec',t}^{y_t}\filtmeas_{\theta',t-1}^{y_{0:t-1}}\eqsp,
\end{align*}
where 
$$
\Delta_{t,s}(y_{s:t}) = F_{\parvec,t}^{y_t}\cdots F_{\parvec,s+1}^{y_{s+1}}F_{\parvec,s}^{y_{s}}\filtmeas_{\theta',s-1}^{y_{0:s-1}} - F_{\parvec,t}^{y_t}\cdots F_{\parvec,s+1}^{y_{s+1}}\filtmeas_{\theta',s}^{y_{0:s}}
$$
with the convention $F_{\parvec,0}^{y_{0}}\filtmeas_{\theta',-1}^{y_{-1}} =\filtmeas_{\theta,0}^{y_{0}}$. Consider also the backward function $\beta_{s|t}^{y_{s+1:t}}$ and the forward smoothing kernel $F_{s|t,\parvec}^{y_{s:t}}$ defined by
\begin{align*}
    \beta_{\parvec,s|t}^{y_{s+1:t}}(x_s) &= \int M_\parvec(x_s,\rmd x_{s+1})g_\parvec^{y_{s+1}}(x_{s+1})\cdots M_\parvec(x_{t-1},\rmd x_{t})g_\parvec^{y_{t}}(x_{t})\eqsp,\\
    F_{\parvec,s|t}^{y_{s:t}}(x_{s-1},x_s) &= \frac{\beta_{s|t}^{y_{s+1:t}}(x_s)m_\parvec(x_{s-1},x_{s})g_\parvec^{y_{s}}(x_{s})}{\int \beta_{s|t}^{y_{s+1:t}}(x)M_\parvec(x_{s-1},\rmd x)g_\parvec^{y_{s}}(x)}\eqsp.
\end{align*}
Following for instance \cite[Chapter~4]{10.5555/1088883}, we can write for all probability measure $\gamma$,
$$
F_{\parvec,t}^{y_t}\cdots F_{\parvec,s+1}^{y_{s+1}}\gamma = \gamma_{\parvec,s|t} F_{\parvec,s+1|t}^{y_{s+1:t}}\cdots F_{\parvec,t|t}^{y_{t}}\eqsp,
$$
where $\gamma_{\parvec,s|t} \propto \beta_{\parvec,s|t}^{y_{s+1:t}}\gamma$. 
Therefore,
$$
\filtmeas_{\theta,t}^{y_{0:t}} - \filtmeas_{\theta',t}^{y_{0:t}} = \sum_{s=0}^{t-1}\left(\gamma_{\parvec,\parvec',s|t} F_{\parvec,s+1|t}^{y_t}\cdots F_{\parvec,t|t}^{y_{s+1}} - \tilde\gamma_{\parvec,\parvec',s|t} F_{\parvec,s+1|t}^{y_t}\cdots F_{\parvec,t|t}^{y_{s+1}}\right)  + F_{\parvec,t}^{y_t}\filtmeas_{\theta',t-1}^{y_{0:t-1}} - F_{\parvec',t}^{y_t}\filtmeas_{\theta',t-1}^{y_{0:t-1}}\eqsp,
$$
where $\gamma_{\parvec,\parvec',s|t} \propto \beta_{\parvec,s|t}^{y_{s+1:t}} F_{\parvec,s}^{y_{s}}\filtmeas_{\theta',s-1}^{y_{0:s-1}}$ and $\tilde\gamma_{\parvec,\parvec',s|t} \propto \beta_{\parvec,s|t}^{y_{s+1:t}}\filtmeas_{\theta',s}^{y_{0:s}}$.  Note that by H\ref{assum:strong:mixing}, for all measurable sets $A$,
$$
F_{\parvec,s|t}^{y_{s:t}}(x_{s-1},A) \geq \frac{\udlow}{\udup}\frac{\int \eta_-(\rmd x_s)  \mathds{1}_A(x)\beta_{s|t}^{y_{s+1:t}}(x)g_\parvec^{y_{s}}(x)}{\int \eta_+(\rmd x)\beta_{s|t}^{y_{s+1:t}}(x)g_\parvec^{y_{s}}(x)}\eqsp,
$$
so that
$$
\left\|\gamma_{\parvec,\parvec',s|t} F_{\parvec,s+1|t}^{y_t}\cdots F_{\parvec,t|t}^{y_{s+1}} - \tilde\gamma_{\parvec,\parvec',s|t} F_{\parvec,s+1|t}^{y_t}\cdots F_{\parvec,t|t}^{y_{s+1}}\right\|_{\mathrm{tv}} \leq \epsilon^{t-s}\left\|\gamma_{\parvec,\parvec',s|t}-\tilde \gamma_{\parvec,\parvec',s|t}\right\|_{\mathrm{tv}}\eqsp,
$$
with $\epsilon = 1 - \udlow/\udup$. This yields
$$
\left\|\filtmeas_{\theta,t}^{y_{0:t}} - \filtmeas_{\theta',t}^{y_{0:t}}\right\|_{\mathrm{tv}} \leq \sum_{s=0}^{t-1}\epsilon^{t-s}\left\|\gamma_{\parvec,\parvec',s|t}-\tilde \gamma_{\parvec,\parvec',s|t}\right\|_{\mathrm{tv}} + \left\| F_{\parvec,t}^{y_t}\filtmeas_{\theta',t-1}^{y_{0:t-1}} - F_{\parvec',t}^{y_t}\filtmeas_{\theta',t-1}^{y_{0:t-1}}\right\|_{\mathrm{tv}}\eqsp.
$$
For all bounded measurable functions $h$,
\begin{align*}
\left|\gamma_{\parvec,\parvec',s|t}(h)-\tilde \gamma_{\parvec,\parvec',s|t}(h)\right| &= \left| \frac{\int \beta_{\parvec,s|t}^{y_{s+1:t}}(x_s) F_{\parvec,s}^{y_{s}}\filtmeas_{\theta',s-1}^{y_{0:s-1}}(x_s)h(x_s)\mu(\rmd x_s)}{\int  \beta_{\parvec,s|t}^{y_{s+1:t}}(x_s) F_{\parvec,s}^{y_{s}}\filtmeas_{\theta',s-1}^{y_{0:s-1}}(x_s)\mu(\rmd x_s)}- \frac{\int \beta_{\parvec,s|t}^{y_{s+1:t}}(x_s)\filtmeas_{\theta',s}^{y_{0:s}}(x_s)h(x_s)\mu(\rmd x_s)}{\int \beta_{\parvec,s|t}^{y_{s+1:t}}(x_s)\filtmeas_{\theta',s}^{y_{0:s}}(x_s)\mu(\rmd x_s)}\right|\eqsp,\\
&\leq \delta^{y_{0:t}}_{\parvec,\parvec',1}(h) + \delta^{y_{0:t}}_{\parvec,\parvec',2}(h)\eqsp,
\end{align*}
where
\begin{align*}
\delta^{y_{0:t}}_{\parvec,\parvec',1}(h) &= \frac{\int \beta_{\parvec,s|t}^{y_{s+1:t}}(x_s) \left|F_{\parvec,s}^{y_{s}}\filtmeas_{\theta',s-1}^{y_{0:s-1}}(x_s) - F_{\parvec',s}^{y_{s}}\filtmeas_{\theta',s-1}^{y_{0:s-1}}(x_s) \right|h(x_s)\mu(\rmd x_s)}{\int  \beta_{\parvec,s|t}^{y_{s+1:t}}(x_s) F_{\parvec,s}^{y_{s}}\filtmeas_{\theta',s-1}^{y_{0:s-1}}(x_s)\mu(\rmd x_s)} \eqsp,\\
\delta^{y_{0:t}}_{\parvec,\parvec',2}(h) &= \frac{\int \beta_{\parvec,s|t}^{y_{s+1:t}}(x_s)\filtmeas_{\theta',s}^{y_{0:s}}(x_s)h(x_s)\mu(\rmd x_s)}{\int \beta_{\parvec,s|t}^{y_{s+1:t}}(x_s)\filtmeas_{\theta',s}^{y_{0:s}}(x_s)\mu(\rmd x_s)}\frac{\int  \beta_{\parvec,s|t}^{y_{s+1:t}}(x_s) \left|F_{\parvec,s}^{y_{s}}\filtmeas_{\theta',s-1}^{y_{0:s-1}}(x_s)- F_{\parvec',s}^{y_{s}}\phi_{\theta',s-1}^{y_{0:s-1}}(x_s)\right|\mu(\rmd x_s)}{\int  \beta_{\parvec,s|t}^{y_{s+1:t}}(x_s) F_{\parvec,s}^{y_{s}}\filtmeas_{\theta',s-1}^{y_{0:s-1}}(x_s)\mu(\rmd x_s)}\eqsp.
\end{align*}
Note that for all $x_s\in\Xset$, by H\ref{assum:strong:mixing},
\begin{multline*}
\udlow \int \eta_-(\rmd x_{s+1})g_\parvec^{y_{s+1}}(x_{s+1})\cdots m_\parvec(x_{t-1},x_{t})g_\parvec^{y_{t}}(x_{t})\mu(\rmd x_{s+2:t})\leq\beta_{\parvec,s|t}^{y_{s+1:t}}(x_s)\\
\leq \udup\int \eta_+(\rmd x_{s+1})g_\parvec^{y_{s+1}}(x_{s+1})\cdots m_\parvec(x_{t-1},x_{t})g_\parvec^{y_{t}}(x_{t})\mu(\rmd x_{s+2:t})\eqsp,
\end{multline*}
so that 
\begin{multline*}
\delta^{y_{0:t}}_{\parvec,\parvec',1}(h) + \delta^{y_{0:t}}_{\parvec,\parvec',2}(h) \leq 2\|h\|_\infty  \|F_{\parvec,s}^{y_{s}}\filtmeas_{\theta',s-1}^{y_{0:s-1}}- F_{\parvec',s}^{y_{s}}\filtmeas_{\theta',s-1}^{y_{0:s-1}}\|_{\mathrm{tv}}\frac{\|\beta_{\parvec,s|t}^{y_{s+1:t}}\|_\infty}{\inf_{x\in\Xset}\beta_{\parvec,s|t}^{y_{s+1:t}}(x_s)} \\
\leq 2\frac{\udup}{\udlow}\|h\|_\infty  \|F_{\parvec,s}^{y_{s}}\filtmeas_{\theta',s-1}^{y_{0:s-1}}- F_{\parvec',s}^{y_{s}}\filtmeas_{\theta',s-1}^{y_{0:s-1}}\|_{\mathrm{tv}}\eqsp.
\end{multline*}
For all bounded measurable function $h$,
$$
\left|F_{\parvec,s}^{y_{s}}\filtmeas_{\theta',s-1}^{y_{0:s-1}}h- F_{\parvec',s}^{y_{s}}\filtmeas_{\theta',s-1}^{y_{0:s-1}}h\right|\leq R_1 + R_2\eqsp,
$$
where
\begin{align*}
    R_1 &= \left|\frac{\int \left(m_\parvec(x',x)g_\parvec^{y_s}(x) - m_{\parvec'}(x',x)g_{\parvec'}^{y_s}(x)\right)\filtmeas_{\theta',s-1}^{y_{0:s-1}}(\rmd x')h(x)\mu(\rmd x)}{\int m_\parvec(x',x'')g_\parvec^{y_s}(x'')\filtmeas_{\theta',s-1}^{y_{0:s-1}}(\rmd x')\mu(\rmd x'')}\right|\eqsp,\\
    R_2 &= \left|\frac{\int m_{\parvec'}(x',x'')g_{\parvec'}^{y_s}(x'')\filtmeas_{\theta',s-1}^{y_{0:s-1}}(\rmd x')h(x'')\mu(\rmd x'')}{\int m_{\parvec'}(x',x'')g_{\parvec'}^{y_s}(x'')\filtmeas_{\theta',s-1}^{y_{0:s-1}}(\rmd x')\mu(\rmd x'')}\right|\cdot\left|\frac{\int \left(m_{\parvec}(x',x'')g_{\parvec}^{y_s}(x'')- m_{\parvec'}(x',x'')g_{\parvec'}^{y_s}(x'')\right)\filtmeas_{\theta',s-1}^{y_{0:s-1}}(\rmd x')\mu(\rmd x'')}{\int m_{\parvec}(x',x'')g_{\parvec}^{y_s}(x'')\filtmeas_{\theta',s-1}^{y_{0:s-1}}(\rmd x')\mu(\rmd x'')}\right|\eqsp.\\
\end{align*}
By H\ref{assum:strong:mixing}-\ref{assum:lip} and Lemma~\ref{lem:bound:filter},
$$
R_1 \leq \frac{\udup}{\udlow c_-(y_s)}\left\{\frac{1}{\udlow c_-(y_{s-1})}\eta_+\otimes\mu\left(\bar g^{y_{s-1}}\otimes \bar g^{y_{s}}\cdot M\right) +  \eta_+(G^{y_s})\right\}\|\parvec-\parvec'\|_2\|h\|_\infty
$$
The same upper bound can be obtained for $R_2$ as
$$
\left|\frac{\int m_{\parvec'}(x',x'')g_{\parvec'}^{y_s}(x'')\filtmeas_{\theta',s-1}^{y_{0:s-1}}(\rmd x')h(x'')\mu(\rmd x'')}{\int m_{\parvec'}(x',x'')g_{\parvec'}^{y_s}(x'')\filtmeas_{\theta',s-1}^{y_{0:s-1}}(\rmd x')\mu(\rmd x'')}\right| \leq \|h\|_\infty\eqsp.
$$
This yields 
$$
\|F_{\parvec,s}^{y_{s}}\filtmeas_{\theta',s-1}^{y_{0:s-1}}- F_{\parvec',s}^{y_{s}}\filtmeas_{\theta',s-1}^{y_{0:s-1}}\|_{\mathrm{tv}} \leq \frac{2\udup}{\udlow c_-(y_s)}\left\{\frac{1}{\udlow c_-(y_{s-1})}\eta_+\otimes\mu\left(\bar g^{y_{s-1}}\otimes \bar g^{y_{s}}\cdot M\right) +  \eta_+(G^{y_s})\right\}\|\parvec-\parvec'\|_2\eqsp, 
$$
which concludes the proof.
\end{proof}

\section{Checking assumptions}
\label{app:C}

\label{sec:check:assum}
In this section, we provide additional assumptions on the state space model and on the variational family to support that our assumptions can be verified.

\begin{hypA}
    \label{assum:strongmixingbis}
    There exist  constants $0 < \udlow < \udup < \infty$ such that for all  $x\in \Xset$,
$$
\udlow \leq \zeta(x) \leq \udup 
$$
and for all $\theta\in\Theta$, $x,x'\in\Xset$,
$$
\udlow \leq m_{\parvec}(x, x') \leq \udup \eqsp.
$$ 
For all  $y_{0:T}\in\Yset^{T+1}$, there exist $\udlowvar^{y_{0:T}}>0$ and $\udupvar^{y_{0:T}}>0$ such that for all $\varphi\in\Phi$, $t\geq 0$, all $x,x'\in \Xset$,
$$ \udlowvar^{y_{0:T}} \leq  q_{\varphi,t\vert t+1}^{y_{0:T}}(x, x') \leq \udupvar^{y_{0:T}}\eqsp. 
$$ 
In addition, for all $\varphi\in\Phi$, all $y_{0:T}\in \Yset^{T+1}$, and all $x\in\Xset$,
$$
\udlowvar^{y_{0:T}} \leq  q_{\varphi,T}^{y_{0:T}}(x) \leq \udupvar^{y_{0:T}}\eqsp.
$$
\end{hypA}
 Assumption A\ref{assum:strongmixingbis} is known as a strong-mixing assumption and allows to verify H\ref{assum:strong:mixing}. It is classical to obtain quantitative bounds on  approximation of joint smoothing distributions, see for instance  \cite{10.3150/07-BEJ6150,10.3150/21-BEJ1431}. It  typically
requires the state space $\Xset$ to be compact. In settings where the bacwkard variartional kernels are Gaussian and obtained with neural networks which are uniformly bounded with respect to the time index and the observations, $\udupvar^{y_{0:T}}$ and $\udlowvar^{y_{0:T}}$ do not depend on the observations.
\begin{hypA}
\label{assum:bound:likelihoodbis}
For all $y\in\Yset$, $\mathrm{inf}_{\theta\in\Theta}\int g_\theta^y(x)\mu(\rmd x) = c_-(y)>0$ and $\mathrm{sup}_{\theta\in\Theta}\int g_\theta^y(x) \mu(\rmd x) = c_+(y)<\infty$. 
\end{hypA}

Lemma~\ref{lem:bound:filter2}, Lemma~\ref{lem:bound:addfunc} and Proposition~\ref{prop:checkH4} allow to obtain explicit constants in H\ref{assum:boundH}. We prove that the functions $h_{t,\theta,\varphi}^{y_{0:T}}$ are upper-bounded explicitly,  and that $\phi_{\theta,t}^{y_{0:t}}$ and $b_{\theta,t-1|t}^{y_{0:t-1}}$ are lower and upper-bounded explicitly, in particular with respect to the observation sequence. 

When the observation space is compact we can also obtain a uniform control with respect to the observations of these quantities which is crucial to check H\ref{assum:moments} and H\ref{assum:pdata}.

\begin{lemma}
\label{lem:bound:filter2}
Assume that A\ref{assum:strongmixingbis} and A\ref{assum:bound:likelihood} hold. For all $\theta\in\Theta$,  all $t\geq 0$, all $y_{0:t}$, $x_t$, 
$$
\frac{\sigma_- g^{y_t}_\theta(x_t)}{\sigma_+c_+(y_t)}\leq \phi^{y_{0:t}}_{\theta,t}(x_{t})\leq \frac{\sigma_+ g^{y_t}_\theta(x_t)}{\sigma_-c_-(y_t)}\,.
$$
\end{lemma}
\begin{proof}
At time 0, we have $\filtdens^{y_{0}}_{\theta,0}( x_{0}) \propto \zeta( x_0)g^{y_0}_\theta(x_0)$ so that by A\ref{assum:strongmixingbis}-\ref{assum:bound:likelihoodbis},
$$
\frac{\sigma_- g^{y_0}_\theta (x_0)}{\sigma_+c_+(y_0)}\leq \filtdens^{y_{0}}_{\theta,0}(x_0)\leq \frac{\sigma_+ g^{y_0}_\theta (x_0)}{\sigma_-c_-(y_0)}\eqsp.
$$
Similarly, 
$$
\filtdens^{y_{0:t}}_{\theta,t}(x_{t}) \propto g^{y_t}_\theta(x_t)\int \filtmeas^{y_{0:t-1}}_{\theta,t-1}(\rmd x_{t-1})m_\theta(x_{t-1}, x_t)\mu(\rmd x_t)\,,
$$
so that by A\ref{assum:strongmixingbis} and A\ref{assum:bound:likelihoodbis},
$$
\frac{\sigma_- g^{y_t}_\theta (x_t)}{\sigma_+c_+(y_t)}\leq \filtdens^{y_{0:t}}_{\theta,t}(x_t)\leq \frac{\sigma_+ \eta g^{y_t}_\theta (x_t)}{\sigma_-c_-(y_t)}\,.
$$
\end{proof}

\begin{lemma}
\label{lem:bound:addfunc}
Assume that A\ref{assum:strongmixingbis} and A\ref{assum:bound:likelihood} hold. For all $\theta$, all $1\leq t\leq T$, all $y_{0:T}$, $x_{t-1}$, $x_t$,
$$
\frac{\sigma_-^2 g^{y_{t-1}}_\theta(x_{t-1})}{\sigma_+^2c_+(y_{t-1})}\leq b^{y_{0:t-1}}_{\theta,t-1|t}(x_t,x_{t-1}) \leq \frac{\sigma_+^2 g^{y_{t-1}}_\theta(x_{t-1})}{\sigma_-^2c_-(y_{t-1})}
$$
and for $1\leq t \leq T-1$, 
\begin{multline*}
\|h^{y_{0:T}}_{t,\theta,\varphi}\|_\infty \leq
 |\log\udlowvar(y_{0:T})|\vee |\log\udupvar(y_{0:T})| \\
 + \sup_{x_{t-1}\in\Xset}\left|\log \frac{\sigma_-^2 c_-(y_{t-1}) \underline g^{y_{t-1}}(x_{t-1})}{\sigma_+^2 c_+(y_{t-1})}\right|\vee \left|\log \frac{\sigma_+^2 c_+(y_{t-1})\bar g^{y_{t-1}}(x_{t-1})}{\sigma_-^2 c_-(y_{t-1})}\right|
\end{multline*}
and 
\begin{multline*}
\|h^{y_{0:T}}_{T,\theta,\varphi}\|_\infty \leq |\log 2\udlowvar^{y_{0:T}}|\vee |\log 2\udupvar^{y_{0:T}}| + \sup_{x_T\in\Xset}\left|\log \frac{\sigma_- \underline g^{y_T}(x_T)}{\sigma_+c_+(y_T)}\right| \vee\left|\log \frac{\sigma_+ \bar g^{y_T}(x_T)}{\sigma_-c_-(y_T)}\right| \\
 + \sup_{x_{T-1}\in\Xset}\left|\log \frac{\sigma_-^2 c_-(y_{T-1}) \underline g^{y_{T-1}}(x_{T-1})}{\sigma_+^2 c_+(y_{T-1})}\right|\vee \left|\log \frac{\sigma_+^2 c_+(y_{T-1})\bar g^{y_{T-1}}(x_{T-1})}{\sigma_-^2 c_-(y_{T-1})}\right|\eqsp,
\end{multline*}
where $h_{t,\theta,\varphi}$, $1\leq t\leq T$,  are defined in \eqref{eq:def:addfunc}.
\end{lemma}
\begin{proof}
By Lemma~\ref{lem:bound:filter},
$$
\frac{\sigma_-^2 g^{y_{t-1}}_\theta(x_{t-1})}{\sigma_+c_+(y_{t-1})}\leq \phi_{\theta,t-1}^{y_{0:t-1}}(x_{t-1})m_\theta(x_{t-1},x_t)\leq \frac{\sigma_+^2 g^{y_{t-1}}_\theta(x_{t-1})}{\sigma_-c_-(y_{t-1})}\eqsp.
$$
Since
$$
b^{y_{0:t-1}}_{\theta,t-1|t}(x_t,x_{t-1}) = \frac{\phi_{\theta,t-1}^{y_{0:t-1}}(x_{t-1})m_\theta(x_{t-1},x_t)}{\int \phi_{\theta,t-1}^{y_{0:t-1}}(x_{t-1})m_\theta(x_{t-1},x_t)\mu(\rmd x_{t-1})}
$$
we get
$$
\frac{\sigma_-^2 c_-(y_{t-1}) g^{y_{t-1}}_\theta(x_{t-1})}{\sigma_+^2 c_+(y_{t-1})}\leq b^{y_{0:t-1}}_{\theta,t-1|t}(x_t,x_{t-1}) \leq \frac{\sigma_+^2 c_+(y_{t-1})g^{y_{t-1}}_\theta(x_{t-1})}{\sigma_-^2 c_-(y_{t-1})}.
$$
%
Now by \eqref{eq:def:addfunc}, for $1\leq t \leq T-1$,  $h^{y_{0:T}}_{t,\theta,\varphi}(x_{t-1},x_t) = \log q^{y_{0:T}}_{\varphi,t-1|t}(x_t,x_{t-1}) - \log b^{y_{0:t-1}}_{\theta,t-1|t}(x_t,x_{t-1})$ 
so that
\begin{multline*}
\left|h^{y_{0:T}}_{t,\theta,\varphi}(x_{t-1},x_t)\right| \leq 
 |\log\udlowvar(y_{0:T})|\vee |\log\udupvar(y_{0:T})| \\
 + \left|\log \frac{\sigma_-^2 c_-(y_{t-1}) g^{y_{t-1}}_\theta(x_{t-1})}{\sigma_+^2 c_+(y_{t-1})}\right|\vee \left|\log \frac{\sigma_+^2 c_+(y_{t-1})g^{y_{t-1}}_\theta(x_{t-1})}{\sigma_-^2 c_-(y_{t-1})}\right|\eqsp,
\end{multline*}
which concludes the proof. In addition,
using that
$$
h^{y_{0:T}}_{T,\theta,\varphi}(x_{T-1},x_T) = \log q^{y_{0:T}}_{\varphi,T-1|T}(x_T,x_{T-1}) - \log b^{y_{0:T-1}}_{\theta,T-1|T}(x_T,x_{T-1}) + \log q^{y_{0:T}}_{\varphi,T}(x_T) - \log \phi^{y_{0:T}}_{\theta,T}(x_T)
$$
yields
\begin{multline*}
\left|h^{y_{0:T}}_{T,\theta,\varphi}(x_{T-1},x_t)\right| \leq 
 |\log 2\udlowvar(y_{0:T})|\vee |\log 2\udupvar(y_{0:T})| + \left|\log \frac{\sigma_- g^{y_T}_\theta(x_T)}{\sigma_+c_+(y_T)}\right| \vee\left|\log \frac{\sigma_+ g^{y_T}_\theta(x_T)}{\sigma_-c_-(y_T)}\right| \\
 + \left|\log \frac{\sigma_-^2 c_-(y_{T-1}) g^{y_{T-1}}_\theta(x_{T-1})}{\sigma_+^2 c_+(y_{T-1})}\right|\vee \left|\log \frac{\sigma_+^2 c_+(y_{T-1})g^{y_{T-1}}_\theta(x_{T-1})}{\sigma_-^2 c_-(y_{T-1})}\right| \eqsp.
\end{multline*}

\end{proof}

\begin{proposition}
\label{prop:checkH4}
Assume that A\ref{assum:strongmixingbis}, A\ref{assum:bound:likelihoodbis} and H\ref{assum:lip} hold. Then H\ref{assum:boundH} holds. More precisely, for all $y_{0:T}\in \Yset^{T+1}$ and all $0\leq t\leq T$,
$$
\sup_{\theta\in\Theta,\varphi\in\Phi}\left\|\int \mu(\rmd x)\left|h^{y_{0:T}}_{t,\theta,\varphi}(x,\cdot)\right|\right\|_\infty = \upsilon_t^{y_{0:T}}<\infty\eqsp,
$$
where $\upsilon_t^{y_{0:T}} = \sup_{\theta\in\Theta,\varphi\in\Phi}\|h^{y_{0:T}}_{t,\theta,\varphi}\|_\infty$ is given in Lemma~\ref{lem:bound:addfunc}. For all $\theta,\theta'\in\Theta$, $\varphi,\varphi'\in\Phi$, $1\leq t\leq T$,
\begin{align*}
    \int \mu\otimes\mu(\rmd x\rmd x')\left|\log q^{y_{0:T}}_{\varphi,t-1|t}(x,x') - \log q^{y_{0:T}}_{\varphi',t-1|t}(x,x')\right| &\leq c_{1,t}^{y_{0:T}}\left\|\varphi-\varphi'\right\|_2\eqsp, \\
    \int \mu\otimes\mu(\rmd x\rmd x')\left|\log b^{y_{0:t-1}}_{\theta,t-1|t}(x,x') - \log b^{y_{0:t-1}}_{\theta',t-1|t}(x,x')\right| &\leq c_{2,t}^{y_{0:t-1}}\left\|\theta-\theta'\right\|_2\eqsp,\\
    \int \mu(\rmd x)\left|\log q^{y_{0:T}}_{\varphi,T}(x) - \log q^{y_{0:T}}_{\varphi',T}(x)\right| &\leq c_{3,T}^{y_{0:T}}\left\|\varphi-\varphi'\right\|_2\eqsp,\\
    \int \mu(\rmd x)\left|\log \phi^{y_{0:T}}_{\theta,T}(x) - \log \phi^{y_{0:T}}_{\theta',T}(x)\right| &\leq c_{4,T}^{y_{0:T}}\left\|\theta-\theta'\right\|_2\eqsp,
\end{align*}
where $c_{1,t}^{y_{0:T}} = (\udlowvar^{y_{0:T}})^{-1}\mu\otimes \mu(K^{y_{0:T}}_{t-1|t})$, $c_{2,t}^{y_{0:t-1}} = 2\udup L_{t-1}(y_{0:t-1})/(\udlow \inf_{x\in\Xset}\underline g^{y_{t-1}}(x))$, $c_{3,t}^{y_{0:T}} = (\udlowvar^{y_{0:T}})^{-1}\mu(K^{y_{0:T}}_{T})$, and $c_{4,T}^{y_{0:T}} = 2\udup  c_+(y_T)L_{T}(y_{0:T})/(\udlow \inf_{x\in\Xset}\underline g^{y_{T}}(x))$.
\end{proposition}

\begin{proof}
For all $\varphi,\varphi'\in\Phi$, $1\leq t\leq T$, 
$$
\left|\log q^{y_{0:T}}_{\varphi,t-1|t}(x,x') - \log q^{y_{0:T}}_{\varphi',t-1|t}(x,x')\right| \leq \frac{\left| q^{y_{0:T}}_{\varphi,t-1|t}(x,x') - q^{y_{0:T}}_{\varphi',t-1|t}(x,x')\right|}{\left|q^{y_{0:T}}_{\varphi,t-1|t}(x,x') \wedge q^{y_{0:T}}_{\varphi',t-1|t}(x,x')\right|}\eqsp,
$$
so that by A\ref{assum:strongmixingbis} and H\ref{assum:lip},
$$
\left|\log q^{y_{0:T}}_{\varphi,t-1|t}(x,x') - \log q^{y_{0:T}}_{\varphi',t-1|t}(x,x')\right| \leq (\udlowvar^{y_{0:T}})^{-1}K^{y_{0:T}}_{t-1|t}(x',x)\|\varphi-\varphi'\|\eqsp,
$$
an we can choose $c_{1,t}^{y_{0:T}} = (\udlowvar^{y_{0:T}})^{-1}\mu\otimes \mu(K^{y_{0:T}}_{t-1|t})$. Similarly, for all $\varphi,\varphi'\in\Phi$, 
$$
\left|\log q^{y_{0:T}}_{\varphi,T}(x) - \log q^{y_{0:T}}_{\varphi',T}(x)\right| \leq \frac{\left| q^{y_{0:T}}_{\varphi,T}(x) - q^{y_{0:T}}_{\varphi',T}(x)\right|}{\left|q^{y_{0:T}}_{\varphi,T}(x) \wedge q^{y_{0:T}}_{\varphi',T}(x)\right|}\eqsp,
$$
so that by A\ref{assum:strongmixingbis} and H\ref{assum:lip},
$$
\left|\log q^{y_{0:T}}_{\varphi,T}(x) - \log q^{y_{0:T}}_{\varphi',T}(x)\right| \leq (\udlowvar^{y_{0:T}})^{-1}K^{y_{0:T}}_{T}(x)\|\varphi-\varphi'\|\eqsp,
$$
and we can choose $c_{3,t}^{y_{0:T}} = (\udlowvar^{y_{0:T}})^{-1}\mu(K^{y_{0:T}}_{T})$. For all $\theta,\theta'\in\Theta$, $1\leq t\leq T$, 
$$
\left|\log b^{y_{0:t-1}}_{\theta,t-1|t}(x,x') - \log b^{y_{0:t-1}}_{\theta',t-1|t}(x,x')\right| \leq \frac{\left| b^{y_{0:t-1}}_{\theta,t-1|t}(x,x') - b^{y_{0:t-1}}_{\theta',t-1|t}(x,x')\right|}{\left|b^{y_{0:t-1}}_{\theta,t-1|t}(x,x') \wedge b^{y_{0:t-1}}_{\theta',t-1|t}(x,x')\right|}\eqsp.
$$
By Lemma~\ref{lem:bound:addfunc},
$$
\left|\log b^{y_{0:t-1}}_{\theta,t-1|t}(x,x') - \log b^{y_{0:t-1}}_{\theta',t-1|t}(x,x')\right| \leq \frac{\sigma_+^2c_+(y_{t-1})}{\sigma_-^2 \underline g^{y_{t-1}}(x_{t-1})}\left| b^{y_{0:t-1}}_{\theta,t-1|t}(x,x') - b^{y_{0:t-1}}_{\theta',t-1|t}(x,x')\right|\eqsp.
$$
Then, noting that $b^{y_{0:t-1}}_{\theta,t-1|t}(x,x') = \phi^{y_{0:t-1}}_{\theta,t-1}(x')m_{\theta}(x',x)/c_\theta(x)$ where $c_\theta(x)  = \int \phi^{y_{0:t-1}}_{\theta,t-1}(x')m_{\theta}(x',x)\mu(\rmd x')$, write 
\begin{multline*}
\left| b^{y_{0:t-1}}_{\theta,t-1|t}(x,x') - b^{y_{0:t-1}}_{\theta',t-1|t}(x,x')\right| \leq \left|\frac{(\phi^{y_{0:t-1}}_{\theta,t-1}(x')-\phi^{y_{0:t-1}}_{\theta',t-1}(x'))m_{\theta}(x',x)}{c_\theta(x)}\right| \\
+ \left|\frac{\phi^{y_{0:t-1}}_{\theta',t-1}(x')(m_{\theta}(x',x)-m_{\theta'}(x',x))}{c_{\theta}(x)}\right| + \left|\frac{\phi^{y_{0:t-1}}_{\theta',t-1}(x')m_{\theta'}(x',x)}{c_{\theta'}(x)}\right|\left|\frac{c_{\theta'}(x)-c_{\theta}(x)}{c_{\theta}(x)}\right|
\end{multline*}
By A\ref{assum:strongmixingbis},
$$
\int \mu\otimes\mu(\rmd x\rmd x')\left|\frac{(\phi^{y_{0:t-1}}_{\theta,t-1}(x')-\phi^{y_{0:t-1}}_{\theta',t-1}(x'))m_{\theta}(x',x)}{\underline g^{y_{t-1}}(x')c_\theta(x)}\right|\leq 2\frac{\udup}{\udlow \inf_{x\in\Xset}\underline g^{y_{t-1}}(x)}\left\|\filtmeas_{\theta,t-1}^{y_{0:t-1}} - \filtmeas_{\theta',t-1}^{y_{0:t-1}}\right\|_{\mathrm{tv}}\eqsp,
$$
and by Lemma
~\ref{lem:bound:addfunc}, we can choose $c_{2,t}^{y_{0:t-1}} = 2\udup L_{t-1}(y_{0:t-1})/(\udlow \inf_{x\in\Xset}\underline g^{y_{t-1}}(x))$. For all $\theta,\theta'\in\Theta$,  
$$
\left|\log \filtdens^{y_{0:T}}_{\theta,T}(x) - \log \filtdens^{y_{0:T}}_{\theta',T}(x)\right| \leq \frac{\left| \filtdens^{y_{0:T}}_{\theta,T}(x) - \filtdens^{y_{0:T}}_{\theta',T}(x)\right|}{\left|\filtdens^{y_{0:T}}_{\theta,T}(x) \wedge \filtdens^{y_{0:T}}_{\theta',T}(x)\right|}\eqsp,
$$
By Lemma~\ref{lem:bound:filter2},
$$
\left|\log \filtdens^{y_{0:T}}_{\theta,T}(x) - \log \filtdens^{y_{0:T}}_{\theta',T}(x)\right| \leq \frac{\sigma_+c_+(y_T)}{\sigma_- \underline g^{y_T}(x)}\left| \filtdens^{y_{0:T}}_{\theta,T}(x) - \filtdens^{y_{0:T}}_{\theta',T}(x)\right|\eqsp.
$$
Therefore,
$$
\int \mu(\rmd x)\left|\log \filtdens^{y_{0:T}}_{\theta,T}(x) - \log \filtdens^{y_{0:T}}_{\theta',T}(x)\right| \leq 2\frac{\udup c_+(y_T)}{\udlow \inf_{x\in\Xset}\underline g^{y_{T}}(x)}\left\|\filtmeas_{\theta,T}^{y_{0:T}} - \filtmeas_{\theta',T}^{y_{0:T}}\right\|_{\mathrm{tv}}\eqsp,
$$
and by Lemma~\ref{lem:bound:addfunc}, we can choose $c_{4,T}^{y_{0:T}} = 2\udup  c_+(y_T)L_{T}(y_{0:T})/(\udlow \inf_{x\in\Xset}\underline g^{y_{T}}(x))$.
\end{proof}
 If the observation space is compact, under standard regularity assumptions, all upper bounds can be obtained uniformly with respect to the observations. 
\end{document}

%% file: math_commands.tex
\usepackage{amsmath,amsfonts,bm}
\usepackage{authblk,breakcites}
\usepackage{amssymb,amsthm}
\usepackage{geometry}

\def\eqref#1{equation~\ref{#1}}

\def\1{\bm{1}}

\DeclareMathAlphabet{\mathsfit}{\encodingdefault}{\sfdefault}{m}{sl}
\SetMathAlphabet{\mathsfit}{bold}{\encodingdefault}{\sfdefault}{bx}{n}

\newcommand{\E}{\mathbb{E}}

